\documentclass{scrartcl}

\usepackage{stmaryrd}
\usepackage{xspace}
\usepackage{color}
\usepackage{natbib}
\usepackage{amsthm}
\usepackage{subfig}

\usepackage{amsthm}

\theoremstyle{plain}
\newtheorem{theorem}{Theorem}
\newtheorem{lemma}{Lemma}
\theoremstyle{definition}
\newtheorem{example}{Example}

\usepackage{epigram}
\usepackage{epigram-desc}
\usepackage{epigram-tag}
\usepackage{epigram-enum}

\usepackage{epigram-nat}
\usepackage{epigram-list}
\usepackage{epigram-vec}
\usepackage{epigram-sum}

\usepackage{type-theory}
\usepackage{elaboration}

\newlength{\rulevgap}
\newlength{\ruleheight}
\newlength{\ruledepth}
\newsavebox{\rulebox}
\newlength{\GapLength}

\newcommand{\Rule}[2]{\savebox{\rulebox}[\width][b]                         %
                              {\( \frac{\raisebox{0in} {\( #1 \)}}       %
                                       {\raisebox{-0.03in}{\( #2 \)}} \)}   %
                      \settoheight{\ruleheight}{\usebox{\rulebox}}          %
                      \addtolength{\ruleheight}{\rulevgap}                  %
                      \settodepth{\ruledepth}{\usebox{\rulebox}}            %
                      \addtolength{\ruledepth}{\rulevgap}                   %
                      \raisebox{0in}[\ruleheight][\ruledepth]               %
                               {\usebox{\rulebox}}}
\newcommand{\Axiom}[1]{\savebox{\rulebox}[\width][b]                        %
                               {$\frac{}{\raisebox{-0.03in}{$#1$}}$}        %
                      \settoheight{\ruleheight}{\usebox{\rulebox}}          %
                      \addtolength{\ruleheight}{\rulevgap}                  %
                      \settodepth{\ruledepth}{\usebox{\rulebox}}            %
                      \addtolength{\ruledepth}{\rulevgap}                   %
                      \raisebox{0in}[\ruleheight][\ruledepth]               %
                               {\usebox{\rulebox}}}

\ColourEpigram

\newcommand{\ie}{i.e.\ }
\newcommand{\eg}{e.g.\ }

\title{Elaborating Inductive Definitions}
\author{Pierre-\'{E}variste Dagand \and Conor McBride}
\date{}

\begin{document}

\maketitle

%% * Abstract

\begin{abstract}

%% ** <- Elaboration inductive to code
%% *** -> Reduction of a syntactic artifact to its semantics
%% *** -> Internalize datatypes in type theory
%% *** -> Formal spec
%% **** -> Can be reasonned about

%% ** <- Open up new opportunities 
%% *** -> To implementers
%% **** -> Bootstrapping
%% *** -> To programmers
%% **** -> Generic programming 

We present an elaboration of inductive definitions down to a universe
of datatypes. The universe of datatypes is an internal presentation of
strictly positive families within type theory. By elaborating an
inductive definition -- a syntactic artifact -- to its code -- its
semantics -- we obtain an internalized account of inductives inside
the type theory itself: we claim that reasoning about inductive
definitions could be carried in the type theory, not in the
meta-theory as it is usually the case. Besides, we give a formal
specification of that elaboration process. It is therefore amenable to
formal reasoning too. We prove the soundness of our translation and
hint at its correctness with respect to Coq's \texttt{Inductive}
definitions. The practical benefits of this approach are numerous. For
the type theorist, this is a small step toward bootstrapping, \ie
implementing the inductive fragment in the type theory itself. For the
programmer, this means better support for generic programming: we
shall present a lightweight \texttt{deriving} mechanism, entirely
definable by the programmer and therefore not requiring any extension
to the type theory.

\end{abstract}

%% * Introduction
%% ** <- Dependant data-types
%% *** <- Cf. ICFP paper
%% *** <- Example: List
%% **** <- Datatypes from the ML day
%% **** -> Simple sum-of-product
%% ***** -> Syntax is mere sugar above trivial object
%% ***** -> Parameters add minor complexity

In a dependent type theory, inductive types come in various shapes and
forms. Unsurprisingly, we can define data-types \`{a} la ML,
following the \emph{sum-of-product} recipe: we offer a choice of
constructors and, for each constructor, comes a product of
arguments. An example of such definition is the vintage and classic
\(\List{}\) datatype:
\[
\ListDef
\]
For the working semanticist, this brings fond memory of a golden era:
this syntax has a trivial categorical interpretation in term of
\emph{signature functor}, here \(L_A X = 1 + A \times X\). Without a
second thought, we can brush away the syntax, mapping the syntactic
representations of sum and product to their categorical
counterpart. Handling parameters comes at a minor complexity cost: we
merely parameterize the functor itself, for instance with \(A\) here.

%% *** <- Example: equality-based vector
%% **** <- Grown-up data-types
%% ***** <- Actual dependent types
%% **** -> Distinction Parameter/Index
%% ***** -> Constraints on indices
%% ***** <- Size questions
%% ****** <- Cf. Chung Kil-Hur

However, these data-type definition are child's play in a
dependently-typed setting: they do not let us use any type dependency
in their definition, nor do they let us define inductive dependent 
types. To obtain grown-ups' datatypes, we move to inductive
families~\citep{dybjer:inductive-families}: we first introduce the notion of
\emph{index} that, unlike parameters, we can \emph{constrain} to a
particular value. The archetypal example of an inductive family is the
\(\Vector{}\) datatype, which can understood as a list indexed by its
length:
\[
\VectorDefEquality
\]
In our syntax, we are purposely explicit about constraints on
indices. Indeed, these constraints eventually turn into a statement
using whatever propositional equality the underlying type theory has
to offer.

%% *** <- Example: indexed vector
%% **** <- New bread of definitions
%% ***** <- Predicted by the model
%% ***** -> Want reified back in to syntax
%% **** <- Use indexing information maximally
%% ***** -> Avoid redundancy in definition
%% ****** <- Cf. Brady 
%% ***** <- Crucial in Ornaments
%% ****** <- Cf. ICFP

Alternatively, our model of inductive families~\citep{alti:lics09}
suggests that vectors could be defined by first \emph{matching} over
the index \(n\): if \(n\) is \(\Zero\), then the only possible
constructor is \(\VNil\), otherwise, if \(n\) is \(\Suc[m]\) for some
\(m\), it must be a \(\VCons\) which tail is of length \(m\). We
reflect this definition style -- computing over indices -- by the
following syntax:
\[
\VectorDef
\]
Note that we are here using the Epigram \emph{by} (\(\DoBy\))
gadget~\citep{mcbride.mckinna:view-from-the-left} to perform the case
analysis on \(n\). A pattern-matching
notation~\citep{coquand:pattern-matching,sozeau:equations} could be
used as well. When the patterns are unsurprising, we shall abuse
notation and write a standard pattern match.

This definition style lets us maximally use the information provided
by the indices to structure datatypes. In the constraint-based
definition, we have to store an index \(n'\) against which we
constrain \(n\). In the computation-based definition, we simply match
against the index. Such difference of presentation has been studied by
\citet{brady:index-inductive-families} in the context of various
optimizations on inductive types. In our work on
ornaments~\citep{mcbride:ornament,dagand:fun-orn}, this definition
style is instrumental in structuring our universes of functional
ornament and enables us to lift functions across ornamented types.

%% ** <- Standard approach: Positivity checker
%% *** -> Part of TCB
%% **** <- Non negligible
%% **** <- Or too simplistic 
%% **** -> Cf. Agda bugs
%% **** -> Cf. Coq limitations

Now, we ought to make sure that our language of data-type is correct,
let alone semantically meaningful. Indeed, if we were to accept the
following definition
\Spacedcommand{\Bottom}{\Canonical{Bad}}
\[
\Data{\Bottom}
     {\Param{\Var{A}}{\Set}}
     {\Set}{
\Emit{\Bottom}{\Var{A}}{\Constructor{ex}\: \PiTel{\Var{f}}{\Bottom[\Var{A}] \To \Var{A}}}}
\]
we would make many formal developments a lot easier to prove! To ban
these bogus definitions, theorem provers such as
Agda~\citep{norell:agda} or Coq~\citep{coq} rely on a positivity
checker to ensure that all recursive arguments are in a
strictly positive position.
The positivity checker is therefore part of the trusted computing base
of the theorem prover. Besides, by working on the syntactic
representation of datatypes, it is a non negligible piece of software
that is a common source of frustration: it either stubbornly prevents
perfectly valid definitions -- as it sometimes is the case in Coq --
or happily accepts obnoxious definitions -- as Agda users discover
every so often.

%% ** <- Syntactic presentation of datatypes
%% *** -> Extremely painful (not necessarily hard) to reason about
%% **** <- Example: balanced binary tree
%% *** -> Dotdotdot proofs
%% **** <- Cf. Stone-Harper
%% ***** <- Quote paper
%% **** <- Cf. Matita refinement
%% **** <- Cf. Luo?
%% **** -> Throw in the towel

For the working semanticist, this is an awakening and a rude one:
while reasoning about ML datatypes used to be at a functor away, she
now has to cope with equality constraints and computations on
indices. Most infuriatingly, we have some elegant models for inductive
families but we seem stuck with some clumsy syntactic presentation:
quoting \citet{harper:elaboration}, ``the treatment of datatypes is
technically complex, but conceptually straightforward''. Following
Stone and Harper, most authors~\citep{coen:refinement,luo:utt,
  mcbride:construction-constructor} have no choice but to throw in the
towel and proceed over a ``\ldots''-filled skeleton of inductive
definition. While this does not make these works any less correct, it
makes them hard for the author to get right and for the reader to
understand.

%% ** <- Disconnection with semantics
%% *** <- Example: Coq
%% **** <- DeBruijn criterion
%% **** -> Tactics elaborates to terms
%% **** -> Minimal, trusted type-checker
%% ***** -> Cf. Barras, fail for inductives
%% **** /> No such thing for datatype declarations
%% *** <- Initial algebra semantics
%% *** <- Strictly positive types/families
%% *** -> Correctness?

We attribute these difficulties to the formal gap between the syntax
of inductive definitions and their semantics. While inductive families
have an interpretation in term of strictly positive functors and their
initial algebra, we are unable to leverage this knowledge. Being stuck
with a syntactic artifact, the ghost of the de Bruijn criterion haunts
our type theories: inductive definitions elude the type checker and
must be enforced by a not-so-small positivity checker. Besides, since
the syntax of inductive definitions is hardly amenable to formal
reasoning, we are left wondering if its intended semantics is indeed
always respected. How many inductive skeletons might be hidden in the
dark closet of your favorite theorem prover?

%% ** <- Alternative: reflecting inductives in type theory
%% *** <- Example: pattern functor of balanced binary tree
%% **** <- Baby step away: code in Universe
%% *** -> "Generic programming is just programming"
%% **** <- Cf. Levitation, Dybjer
%% *** -> Example: internalize 'construction on constructors'
%% *** /> Need user-friendly interface
%% **** <- Can we recover the 'data' syntax while generating codes?
%% **** -> Best of both worlds

An alternative to a purely syntactic approach is to reflect inductive
types inside the type theory itself. Following
\citet{benke:universe-generic-prog}, we extend a Martin-L\"{o}f type
theory with a universe of inductive families, all that for a minor
complexity cost~\citep{dagand:levitation}. From within the type
theory, we are then able to create and manipulate datatypes but also
compute over them. However, from a user perspective, these codes are a
no-go: manually coding datatypes, for instance, is too cumbersome. Rather
than writing low-level codes, we would like to write a
honest-to-goodness inductive definition and get the computer to
automatically \emph{elaborate} it to a code in the universe.

%% ** -> Contribution
%% *** -> Propose syntax for datatype declaration
%% **** <- Inspired by Coq and Agda
%% ***** -> Conservative extension
%% **** /> With a twist
%% ***** <- Support index-level computation
%% ***** -> More emphasis on computational aspect
%% **** -> Elaborating on ICFP paper
%% ***** <- Pun

In this paper, we elaborate upon (pun intended) the syntax of
datatypes introduced in \citet{dagand:fun-orn}. In our previous work,
we presented a conservative extension of an \texttt{Inductive}-like
syntax, the twist being in our support of computation over
indices. While this syntax had been informally motivated, this paper
gives a formal specification of its elaboration down to our universe
of datatypes. Our contributions are the following:

\begin{itemize}

%% *** -> Crash course in elaboration
%% **** <- Present few techniques
%% ***** <- Bidirectional type checking
%% ***** <- Types as representation
%% ***** <- Elaboration on high-level construct
%% ****** <- Program def.
%% **** /> Nothing new per se
%% ***** <- Cf. Local type inf, View from the left, Matita
%% ***** /> Integrated into a coherent collection of techniques
%% ***** /> Aiming at a general, type-directed framework

\item In Section~\ref{sec:elab-course}, we give a crash course in
  elaboration for dependent types. We will present a bidirectional
  type checker~\citep{pierce:bidirectional-tc} for our type theory. We
  then extend it to make programming a less cryptic experience. To
  that purpose, we shall use types as \emph{presentations} of more
  high-level concepts, such as the notions of finite set or of
  datatype constructor. While this Section does not contain any new
  result \emph{per se}, we aim at introducing the reader to a coherent
  collection of techniques that, put together, form a general
  framework for type-directed elaboration ;

%% *** -> Elaboration of inductive types to Desc
%% **** <- Whet our appetite
%% **** -> Develop intuition and techniques for families
%% **** -> Captures strictly positive types

\item In Section~\ref{sec:elab-data-types}, we specify the elaboration
  of inductive types down to a simple universe of inductive
  types. While the system we present in this Section is restricted to
  strictly positive types, we take advantage of its simplicity to
  develop our intuition. The same ideas are at play in the case of
  inductive families ;

%% *** -> Elaboration of inductive families to IDesc
%% **** <- Subsumes inductive types
%% ***** /> Build upon it
%% **** -> Supporting index-level computations
%% **** -> Captures strictly positive families

\item In Section~\ref{sec:elab-data-family}, we specify the
  elaboration of inductive families down to our universe of inductive
  families. This system subsumes the one introduced in
  Section~\ref{sec:elab-data-types} but we should reuse much of the
  concepts developed in that Section. The novelty of our syntax is to
  support computation on indices, as made possible by our model of
  inductive families and its universe presentation ;

%% *** (-> Elaboration of ornament declarations?)
%% *** -> Discuss new possibilities
%% **** <- Internalized construction on constructors
%% **** <- Internalized deriving mechanism

\item In Section~\ref{sec:discussion}, we draw the consequences of our
  design choice. For the proof-assistant implementer, we show how
  meta-theoretical results on inductives, such
  as~\citet{mcbride:construction-constructor}, can be internalized and
  formally presented in the type theory. For the programmer, we show
  how a generic \texttt{deriving} mechanism \`{a} la Haskell
  can be implemented from within the type theory.

\end{itemize}

%% *** /> Disclaimer
%% **** <- *Specification* of an elaboration procedure
%% **** <- Not an implementation
%% ***** -> Implementers might be disappointed
%% **** -> Easier to carry formal reasoning
%% ***** -> Our goal in this paper
%% ***** -> Appeal to people working *on* Inductive
%% **** /> Implementation on its way
%% ***** <- Work in progress
%% ***** <- Dealt with Inductive type translation
%% ***** <- Various prototypes for inductive families in Epigram
%% ****** /> Ever changing..

\paragraph{Scope of this work:} 

This paper aims at \emph{specifying} an elaboration procedure from an
inductive definition down to its representation in a universe of
inductive types. At the risk of disappointing implementers, we are
not describing an implementation. In particular, we shall present the
elaboration in a relational style, hence conveniently glancing over
the operational details. Our goal is to ease the formal study of
inductive definitions, hence the choice of this more abstract style.
Nonetheless, this paper is not entirely disconnected from
implementation. First, it grew out of our work on the Epigram
system~\citep{pigs:epigram}, in which Peter Morris implemented a
tactic elaborating an earlier form of inductive definition down to our
universe of code. Second, a tutorial implementation of an elaborator
for inductive types is currently underway.

%% * A Calculus of Inductive Types

\section{The Type Theory}
\label{sec:type-theory}

%% ** <- Motivation
%% *** <- Recall results from Levitation
%% **** -> Self-contained paper
%% *** <- Universe of SPT/SPF
%% **** -> Internalized notion of datatypes
%% ***** -> Semantic approach to datatypes
%% ***** <- Initial algebra of SPT/SPF
%% *** -> Look elsewhere for meta-theoretical study
%% **** <- Levitation
%% **** <- Luo

For this paper to be self-contained, we shall recall a few definitions
from our previous work~\citep{dagand:levitation}. We shall not dwell
on the meta-theoretical properties of this system: the interested
reader should consult~\citet{luo:utt} for a study of Martin-L\"{o}f
type theory and~\citet{dagand:levitation} for its extension with a
universe of datatypes.

%% ** <- Core type theory
%% *** <- Standard Martin-Lof
%% **** <- Use Set for sort
%% ***** -> Careful Coq users
%% ***** <- Simplicity: assume cumulativity
%% **** <- Sigma, Pi, Unit
%% ***** -> Nothing controversial
%% ***** -> Adaptable to variations
%% **** <- Equality agnostic
%% ***** <- Judgemental *spec*
%% ***** -> Either Coq, Agda, OTT, etc

We present our core type theory in Figure~\ref{fig:type-theory}. It is
a standard Martin-L\"{o}f type theory, with \(\Sigma\)-types and
\(\Pi\)-types. We shall write \(\Set[k]\) for the hierarchy of types,
implicitly assuming cumulativity of universes. In order to be
equality-agnostic, we simply specify our expectations through a
judgmental presentation. Our presentation is hopefully not
controversial and should adapt easily to regional variations, such as
Coq, Agda or an observational type theory~\citep{altenkirch:ott}.

\renewcommand{\EpigramEquality}{
\begin{array}{@{}c}
%% Beta-reduction
\Rule{\begin{array}{@{}l}
        \TypeJudgment{\Gamma}{S}{\Set} \quad
        \TypeJudgment{\Gamma ; \XS}{t}{T} \\
        \TypeJudgment{\Gamma}{s}{S}
      \end{array}}
     {\EqualJudgment{\Gamma}{(\LamAnn{\X}{S} t)\:s}{t[s/\X]}{T[s/\X]}}
\\
\Rule{\begin{array}{@{}l}
       \TypeJudgment{\Gamma}{s}{S} \quad
       \TypeJudgment{\Gamma ; \XS}{T}{\Set} \\
       \TypeJudgment{\Gamma}{t}{T[s/\X]}
      \end{array}}
     {\EqualJudgment{\Gamma}{\Fst[(\PairAnn{s}{t}{\X.T})]}{s}{S}}
\\
\Rule{\begin{array}{@{}l}
       \TypeJudgment{\Gamma}{s}{S} \quad
       \TypeJudgment{\Gamma ; \XS}{T}{\Set} \\
       \TypeJudgment{\Gamma}{t}{T[s/\X]}
      \end{array}}
     {\EqualJudgment{\Gamma}{\Snd[(\PairAnn{s}{t}{\X.T})]}{t}{T[s/\X]}}
\end{array}
}

\renewcommand{\EpigramContextValidity}{
    %% Empty context validity
    \Axiom{\ContextValid{}}
    \\
    %% Extend context
    \Rule{\ContextValid{\Gamma} \quad
          \TypeJudgment{\Gamma}{\Meta{S}}{\Set[k]}}
         {\ContextValid{\Gamma ; \XS}}\;\X\not\in\Gamma
}

\begin{figure}[tbp]
\centering
\begin{tabular}{l@{\qquad}l}
\begin{tabular}{c}
\subfloat[][Context validity]
         {
{\small
\(\Code[c]{
\EpigramContextValidity \\
\Rule{\ContextValid{\Gamma} \quad
      \Code[c]{
        \TypeJudgment{\Gamma}{\Meta{S}}{\Set[k]} \\
        \TypeJudgment{\Gamma}{\Meta{t}}{\Meta{S}}}}
     {\ContextValid{\Gamma ; \TypeAnn{\X \mapsto \Meta{t}}{\Meta{S}}}}\;\X\not\in\Gamma
}\)}}  \\
\subfloat[][Judgmental equality]{\small \(\EpigramEquality\)}
\end{tabular}
&
\subfloat[][Typing judgments]{\small \(\EpigramTypeSystem\)}
\end{tabular}

\caption{Type theory}
\label{fig:type-theory}

\end{figure}

%% ** <- Universes of finite sets
%% *** -> Motivation
%% **** <- Named finite collections
%% ***** -> Enumeration a la C
%% *** -> Tag
%% *** -> EnumU, EnumT
%% *** -> Finite pi
%% *** <- Syntactic sugar
%% **** <- {`a, `b, `c}
%% ***** -> Example: translate to code
%% **** <- {`a -> e, `b -> f, `c -> g }
%% ***** -> Example: translate to code

\Spacedcommand{\spi}{\Function{\(\pi\)}}

A first addition to this core calculus is a universe of
enumerations~(Fig.~\ref{fig:enum-universe}). The purpose of this
universe is to let us define finite collections of labels. Labels are
introduced through the \(\UId\) type and are then used to define
finite sets through the \(\EnumU\) universe. To index a specific
element in such a set, we write an \(\EnumT\) code. To eliminate
finite sets, we form a small \(\Pi\)-type \( \TypeAnn{\spi}
{\PiTel{\Var{E}}{\EnumU} \PiTo{\Var{P}}{\EnumT[\Var{E}] \To \Set}
  \Set} \) that builds a lookup tuple mapping, for all label \(e\) in
the enumeration, a value of type \(P\: e\). To perform the lookup, we
define the eliminator:
\[ \TypeAnn{\EnumElim} {\PiTel{\Var{E}}{\EnumU}
  \PiTel{\Var{P}}{\EnumT[\Var{E}] \To \Set} \To \spi[\Var{E}\:
    \Var{P}] \To \PiTel{\Var{x}}{\EnumT[\Var{E}]} \To \Var{P}\: \X}
\]

\begin{example}[Coding \(\Collection{\Tag{a}, \Tag{b}, \Tag{c}}\)]

We define this set by merely enumerating its labels, in effect
building a list of tags:
\[
\Collection{\Tag{a}, \Tag{b}, \Tag{c}} \triangleq \TypeAnn{\ConsEnum[\Tag{a}\: (\ConsEnum[\Tag{b}\: (\ConsEnum[\Tag{c}\: \NilEnum])])]}{\EnumU}
\]

\end{example}

\begin{example}[Coding \(\TypeAnn{\CollectionElim{\Tag{a} \mapsto e_a, \Tag{b} \mapsto e_b, \Tag{c} \mapsto e_c}}
                                 {\PiTo{\Var{x}}{\EnumT[\Collection{\Tag{a}, \Tag{b}, \Tag{c}}]} P\: x}\)]

We define this function by a straightforward application of the
\(\EnumElim\) eliminator:
\[
\CollectionElim{\Tag{a} \mapsto e_a, \Tag{b} \mapsto e_b, \Tag{c} \mapsto e_c}
    \triangleq
        \EnumElim[(\ConsEnum[\Tag{a}\: (\ConsEnum[\Tag{b}\: (\ConsEnum[\Tag{c}\: \NilEnum])])])\:
                  P\:
                  \Pair{e_a}{\Pair{e_b}{\Pair{e_c}{\Void}}}]
\]

\end{example}

\begin{figure}[tb]

{\small
\begin{tabular}{lll}
\subfloat[][Tags]{\(\Code[c]{\EpigramTypeUId\\ \EpigramTypeTag}\) \label{fig:label}} &
\subfloat[][Enumeration]{\(\Code[c]{\EpigramTypeEnumU\\ \EpigramTypeNilEnum\\ \EpigramTypeConsEnum}\)} &
\subfloat[][Index]{\(\Code[c]{\EpigramTypeEnumT\\ \EpigramTypeZeroEnumT\\ \EpigramTypeSucEnumT}\)}
\end{tabular}
}

\caption{Universe of enumerations}
\label{fig:enum-universe}

\end{figure}

%% ** <- Universe of inductive types
%% *** <- Captures ML datatypes
%% **** -> Pedagogical purpose
%% **** <- Prepare intuition for full-blown universe
%% *** <- Idea
%% **** <- Code gives grammar 
%% ***** -> \emph{Desc}ribe functors
%% **** <- Intepretation gives semantics
%% ***** -> Build the functor over Set
%% **** <- Mu takes least fix-point
%% ***** -> Comes with generic elimination principle

We recall the definition of our universe of inductive types in
Figure~\ref{fig:universe-types}. This universe captures
strictly positive types, a generalization of ML datatypes to dependent
types. For pedagogical reason, we choose this simple universe as a
first step toward a full-blown universe of inductive families.
The idea at work behind this presentation is the following: to define
new datatypes, we give their code by \emph{desc}ribing their signature
functor in \(\Desc\). The interpretation function
\(\InterpretDesc{\_}\) turns such a description into the corresponding
endofunctor over \(\Set\). Its definition is obvious from the codes: a
\(\DSigma\) is interpreted into a \(\Sigma\)-type, and so on. The
notable exception is \(\DVar\), which describes the identity
functor. Remark that the resulting functor are strictly positive, by
construction. Hence, we can construct their least fix-point using
\(\Mu\) and safely provide a generic elimination principle,
\(\induction\). The \(\All\) function computes the inductive
hypothesis. The interested reader will find their precise definition
in \citet{dagand:levitation} but the basic intuition we have given
here is enough to understand this paper.

\begin{figure}[tb]

\centering
\subfloat[][Codes]{\small\(
\Code[c]{
  \EpigramTypeDesc \\
  \EpigramTypeDescVar \qquad   \EpigramTypeDescUnit \\
  \EpigramTypeDescTimes \\
  \EpigramTypeDescPi \\
  \EpigramTypeDescSigma \\
  \EpigramTypeDescsigma
}\)}
\qquad
\subfloat[][Fix-point]{\small\(
\Code[c]{
%  \TypeAnn{\InterpretDesc{\_}}{\PiTo{\Var{D}}{\Desc} \Set \To \Set} \\
  \Let{\InterpretDesc{\PiTel{\Var{D}}{\Desc}} & \PiTel{\Var{X}}{\Set}}{\Set}{}
  \\
  \Case{
    \Return{\InterpretDesc{\DVar} & \Var{X}}
           {\Var{X}}
    \Return{\InterpretDesc{\DUnit} & \Var{X}}
           {\Unit}
    \Return{\InterpretDesc{\Var{A} \DTimes \Var{B}} & \Var{X}}
           {\InterpretDesc{\Var{A}}[\Var{X}] \Times \InterpretDesc{\Var{B}}[\Var{X}]}
    \Return{\InterpretDesc{\DPi[\Var{S}\: \Var{T}]} & \Var{X}}
           {\PiTo{\Var{s}}{\Var{S}}{\InterpretDesc{\Var{T}\: \Var{s}}[\Var{X}]}}
    \Return{\InterpretDesc{\DSigma[\Var{S}\: \Var{T}]} & \Var{X}}
           {\SigmaTimes{\Var{s}}{\Var{S}}{\InterpretDesc{\Var{T}\: \Var{s}}[\Var{X}]}}
    \Return{\InterpretDesc{\Dsigma[\Var{E}\: \Var{T}]} & \Var{X}}
           {\SigmaTimes{\Var{e}}{\EnumT[\Var{E}]}{\InterpretDesc{\Var{T}\: \Var{e}}[\Var{X}]}}
  }
  \\
  \\
  \EpigramTypeMu \qquad
  \EpigramTypeCon
}\)}
\\
\subfloat[][Elimination principle]{\small
\(\TypeAnn{\induction} 
        {\PiTel{\Var{D}}{\Desc}
         \PiTo{\Var{P}}{\Mu{\Var{D}} \To \Set}
         (\PiTo{\Var{d}}{\InterpretDesc{\Var{D}}[(\Mu{\Var{D}})]}
          \All[\Var{D}\: (\Mu{\Var{D}})\: \Var{P}\: \Var{d}] \To \Var{P} (\In{\Var{d}})) \To
          \PiTo{\Var{x}}{\Mu{\Var{D}}} \Var{P} \Var{x}}\)}

\caption{Universe of inductive types}
\label{fig:universe-types}

\end{figure}

%% **** <- Examples
%% ***** <- Nat
%% ***** <- Tree

\begin{example}[Describing natural numbers]

\renewcommand{\NatDDef}{
  \Let{\NatD}{\Desc}{
    \Return{\NatD}
           {\Dsigma[\Collection{
                        \begin{array}{l}
                          \Tag{0},\\ 
                          \Tag{\CN{suc}}
                        \end{array}}\:
                    \CollectionElim{
                      \begin{array}{l@{\DoReturn}l}
                        \Tag{0} & \DUnit, \\
                        \Tag{\CN{suc}} & \DVar
                      \end{array}
                    }]}}
}

Natural numbers can be presented as the least fix-point of the
functor \(F X = 1 + X\). This is described by:
\[
\NatDDef 
\qquad
\Let{\Nat}{\Set}{
  \Return{\Nat}
         {\Mu[\NatD]}
}
\]
Note that we are using a small \(\Dsigma\) here: a datatype definition
always starts with a finite choice of constructors. This corresponds
to the \emph{tagged descriptions} of \citet{dagand:levitation}: using
this structure, we can implement some generic constructions -- such as
the Zipper or the free monad -- and generic theorems -- such as in
Section~\ref{sec:const-on-const}.

\end{example}

\newcommand{\TreeD}{\Function{TreeD}}

\begin{example}[Describing binary trees]

Categorically, binary trees are modeled by the least fix-point of the
functor \(T_A X = 1 + A \times X^2\). In our universe, we obtain trees by the following definitions:
\[
\Code{
  \Let{\TreeD & \PiTel{\Var{A}}{\Set}}{\Desc}{
    \Return{\TreeD & \Var{A}}
           {\Dsigma[\Collection{
                        \begin{array}{l}
                          \Tag{\CN{leaf}},\\ 
                          \Tag{\CN{node}}
                        \end{array}}\:
                    \CollectionElim{
                      \begin{array}{l@{\DoReturn}l}
                        \Tag{\CN{leaf}} & \DUnit, \\
                        \Tag{\CN{node}} & \DVar \DTimes \DSigma[\Var{A}\: \Lam{\_} \DVar]
                      \end{array}
                    }]}}
\\
\Let{\Function{Tree} & \PiTel{\Var{A}}{\Set}}{\Set}{
  \Return{\Function{Tree} & \Var{A}}
         {\Mu[(\TreeD\: \Var{A})]}
}}
\]
\end{example}

%% ** <- Universe of inductive families

%% *** <- Subsumes previous universe
%% **** -> Our actual target
%% *** -> Capture inductive families
%% **** <- Parameters
%% **** <- Indexes
%% ***** /> Particularity: computation on indices
%% ***** -> Take advantage of it

Finally, we recall the universe of indexed descriptions
(Fig.~\ref{fig:universe-families}), which captures inductive
families. This universe subsumes the previous one by not only allowing
parameters but also indices: thanks to this universe, we can describe
datatypes such as vector. Just as before, we give a universe of codes,
\(\IDesc\), that are then interpreted to build the least
fix-point. Note that the argument of \(\IMu\) is a \emph{function}
from indices to codes: this particularity lets us define datatypes by
computation over the index. We illustrate this point through two
possible implementations of the vector datatype.

\renewcommand{\InterpretIDescDef}{
\Code{
  \Let{\InterpretIDesc{\PiTel{\Var{D}}{\IDesc[\Var{I}]}} & 
              \PiTel{\Var{X}}{\Var{I} \To \Set}}
      {\Set}{}\\
\Case{
\Return{\InterpretIDesc{\DVar[\Var{i}]} & \Var{X}}{\Var{X}\: \Var{i}}
\Return{\InterpretIDesc{\DUnit} & \Var{X}}{\Unit}
\Return{\InterpretIDesc{\Var{A} \DTimes \Var{B}} & \Var{X}}
       {\InterpretIDesc{\Var{A}}[\Var{X}] \Times \InterpretIDesc{\Var{B}}[\Var{X}]}
\Return{\InterpretIDesc{\DPi[\Var{S}\: \Var{T}]} & \Var{X}}{\PiTo{\Var{s}}{\Var{S}}{\InterpretIDesc{\Var{T}\: \Var{s}}[\Var{X}]}}
\Return{\InterpretIDesc{\DSigma[\Var{S}\: \Var{T}]} & \Var{X}}{\SigmaTimes{\Var{s}}{\Var{S}}{\InterpretIDesc{\Var{T}\: \Var{s}}[\Var{X}]}}
\Return{\InterpretIDesc{\Dsigma[\Var{E}\: \Var{T}]} & \Var{X}}
       {\SigmaTimes{\Var{e}}{\EnumT[\Var{E}]}
                   {\InterpretIDesc{\Var{T}\: \Var{e}}[\Var{X}]}}
}}}

\begin{figure}[tbp]

\centering
\subfloat[][Codes]{\small \(
\Code[c]{
  \EpigramTypeIDesc \\
  \EpigramTypeIDescVar \\
  \EpigramTypeIDescUnit \\
  \EpigramTypeIDescTimes \\
  \EpigramTypeIDescPi \\
  \EpigramTypeIDescSigma \\
  \EpigramTypeIDescsigma 
}\)}
\qquad
\subfloat[][Fix-point]{\small \(
\Code[c]{
  \InterpretIDescDef
  \\
  \\
  \EpigramTypeIMu \\
  \EpigramTypeICon
}\)}
\\
\subfloat[][Elimination principle]{
\small
\(
\begin{array}{l@{\:\:}l}
\TypeAnn{\Iinduction} 
        {& \PiTel{\Var{R}}{\Var{I} \To \IDesc[\Var{I}]}
           \PiTo{\Var{P}}{(\SigmaTimes{\Var{i}}{\Var{I}}
                                      {\IMu[\Var{R}\: \Var{i}]}) \To \Set} \\
         & (\PiTel{\Var{i}}{\Var{I}} 
           \PiTo{\Var{xs}}{\InterpretIDesc{\Var{R}\: \Var{i}}[(\IMu[\Var{R}])]}
           \IAll[(\Var{R}\: \Var{i})\:
                 (\IMu[\Var{R}])\:
                 \Var{xs}\:
                 \Var{P}] \To
           \Var{P}\: \Pair{\Var{i}}{\In[\Var{xs}]}) \To  \\
         & \PiTel{\Var{i}}{\Var{I}}
           \PiTo{\Var{x}}{\IMu{\Var{I}}{\Var{R}}{\Var{i}}}
           \Var{P}\: \Pair{\Var{i}}{\Var{x}}}
\end{array}
\)
}

\caption{Universe of inductive families}
\label{fig:universe-families}

\end{figure}

\begin{example}[Vector, \`{a} la Coq]
  
  The standard presentation of inductive families is by using equality
  constraints to force the indexing strategy. In effect, this
  corresponds to the following definition of vectors:
  \[\Code{
  \Let{\VectorD{} & 
           \PiTel{\Var{A}}{\Set} & 
           \PiTel{\Var{n}}{\Nat}}
      {\IDesc[\Nat]}{
        \Return{\VectorD{} & \Var{A} & \Var{n}}
               {\begin{array}[t]{l@{\:}l}
                   \Dsigma & 
                        \Collection{
                            \Code{\Tag{\CN{nil}},\\
                                  \Tag{\CN{cons}}}}\\
                           &
                        \CollectionElim{
                            \Code{\Tag{\CN{nil}} \mapsto \DSigma[(\Var{n} \PropEqual \Zero)\: \Lam{\_}{\DUnit}], \\
                                  \Tag{\CN{cons}} \mapsto \DSigma[\Nat\: 
                                                       \Lam{\Var{m}}{\DSigma[(\Var{n} \PropEqual (\Suc[\Var{m}]))\: 
                                                         \Lam{\_}{\DSigma[\Var{A}\: 
                                                           \Lam{\_}{\DVar[\Var{m}]}]}]}]}}
                 \end{array}}}
  \\
  \Let{\Vector{} & 
         \PiTel{\Var{A}}{\Set} &
         \PiTel{\Var{n}}{\Nat}}
      {\Set}{
        \Return{\Vector{} & \Var{A} & \Var{n}}
               {\IMu[(\VectorD{}\: \Var{A})\: \Var{n}]}
      }
}
  \]

\end{example}

\begin{example}[Vector, alternatively]

  Interestingly, in the previous definition, we are defining
  \(\Vector{D}\) as a function from \(\PiTel{\Var{n}}{\Nat}\) to
  \(\IDesc[\Nat]\) but we are not making use of our ability to inspect
  \(\Var{n}\): we merely constrain it to be what we wish it was. We
  can make full use of that function space by \emph{first} checking
  whether \(\Var{n}\) is \(\Zero\) or \(\Suc[\Var{m}]\) and \emph{only
    then} give a code for the corresponding constructor, respectively
  \(\VNil\) and \(\VCons\):
  \[
  \Let{\VectorD{} & 
           \PiTel{\Var{A}}{\Set} & 
           \PiTel{\Var{n}}{\Nat}}
      {\IDesc[\Nat]}{
        \Return{\VectorD{} & \Var{A} & \Zero}{\DUnit}
        \Return{\VectorD{} & \Var{A} & (\Suc[m])}{\DSigma[\Var{A}\: 
                                                           \Lam{\_}{\DVar[\Var{m}]}]}}
  \qquad
  \Let{\Vector{} & 
         \PiTel{\Var{A}}{\Set} &
         \PiTel{\Var{n}}{\Nat}}
      {\Set}{
        \Return{\Vector{} & \Var{A} & \Var{n}}
               {\IMu[(\VectorD{}\: \Var{A})\: \Var{n}]}
      }
  \]

\end{example}

%% * First steps in elaboration

\section{First Steps in Elaboration}
\label{sec:elab-course}

In this Section, we shall make our first steps in elaboration. First,
we present a bidirectional type system for the calculus introduced in
the previous section. Using the flow of information from type
synthesis to type checking, we are then able to declutter our term
language. Secondly, we adapt the notion of programming
problem~\citep{mcbride.mckinna:view-from-the-left} to our system and
we shall see how, with little effort, we can move away from an austere
calculus and closer to a proper programming language.

%% ** <- Bidirectional type-checking

\subsection{Bidirectional type checking}

%% *** <- Cf. Levitation
%% *** <- Cf. also Coen
%% *** -> Idea: Use flow of information
%% **** -> Push type into term
%% ***** <- Top-level type definition flows locally
%% ***** -> Guide representation of terms
%% ****** -> Allow synthesis of type in abstraction
%% ******* <- Write untyped code
%% ****** -> Extended with Sigma as tuples 
%% ******* <- Write LISP tuples

The idea of bidirectional type
checking~\citep{pierce:bidirectional-tc} is to capture, in the
specification of the type checker, the local flow of typing
information. On the one hand, we will \emph{synthesize} types from
variables and functions while, on the other hand, we will \emph{check}
terms against these synthesized types. By checking terms against their
types, we can use types to structure the term language: for example,
we remove the need for writing the domain type in an abstraction, or
we can use a tuple notation for telescopes of
\(\Sigma\)-types. Following \citep{harper:elaboration}, we distinguish
two categories of terms. The core type theory defines the
\emph{internal} language. The bidirectional approach lets us then
extend this core language into a more convenient \emph{external}
language.
We shall hasten to add that using a bidirectional approach in a
dependently-typed framework is not a novel idea: for instance, it is
the basis of Matita's refinement system~\citep{coen:refinement}. Also,
we gave a similar presentation in some earlier
work~\citep{dagand:levitation}.

\begin{figure}[tb]

\centering
\begin{tabular}{cc}
\subfloat[][Type synthesis]{
\(\Code[c]{
\boxed{\ElabInfer{\Gamma}{t}{t'}{T}}
\\
\\
\EpigramElabInfer
}\) 
\label{fig:type-synthesis}} &
\subfloat[][Type checking]{
\(\Code[c]{
\boxed{\ElabCheck{\Gamma}{t}{T}{t'}}
\\
\\
\EpigramElabCheck
}\) \label{fig:type-checking}}
\end{tabular}

\caption{Bidirectional type checker}
\label{fig:bidir-tc}

\end{figure}

%% *** <- Type synthesis
%% **** -> Source of typing information
%% **** -> Initiate the flow
%% ***** -> Notation: input on the left, output on the right

%% \ElabSpec{the high-level term \(t\) elaborates to the type-theoretic
%%   term \(t'\) of type~\(T\)}
%%          {\(\TypeJudgment{\Gamma}{t'}{T}\)}

Let us present type synthesis~(Fig.~\ref{fig:type-synthesis})
first. The judgment \(\ElabInfer{\Gamma}{t}{t'}{T}\) states that the
external term \(t\) elaborates to the internal term \(t'\) of
type \(T\) in the context \(\Gamma\). By convention, we shall keep
inputs to the relation to the left of the \(\ElabRel{}\) symbol, while
outputs will be on the right. We expect the following soundness
property to hold:
\begin{theorem}[Soundness of type synthesis]
If \(\ElabInfer{\Gamma}{t}{t'}{T}\),
then \(\TypeJudgment{\Gamma}{t'}{T}\).
\end{theorem}

%% *** <- Type checking

%% \ElabSpec{the high-level term \(t\) checked against the type \(T\)
%%   elaborates to the type-theoretic term \(t'\)}
%%          {\(\TypeJudgment{\Gamma}{t'}{T}\)}

%% **** -> Consume flow of information
%% ***** -> Notation: input on the left, output on the right

While type synthesis initiates the flow of typing information, type
checking (Fig.~\ref{fig:type-checking}) lets us make use of this
information to enrich the language of terms. The interpretation of the
judgment \(\ElabCheck{\Gamma}{t}{T}{t'}\) is that the external term
\(t\) is checked against the type \(T\) in context \(\Gamma\) and
elaborates to an internal term \(t'\). Again, we expect the following
soundness property to hold:
\begin{theorem}[Soundness of type checking]
If \(\ElabCheck{\Gamma}{t}{T}{t'}\), then \(\TypeJudgment{\Gamma}{t'}{T}\).
\end{theorem}

%% **** -> Lambda without type
%% **** -> Tag in EnumT

\subsection{Putting types at work}

While the type checker we have specified so far only allows untyped
abstractions, we extend it with some convenient features. 

%% **** -> Tuple notation for Sigma + unit

\paragraph{Tuples:}
By design of the interpretation of our universes of enumeration and
inductive types, we are often going to build inhabitants of
\(\Sigma\)-telescope of the form
\(\SigmaTimes{\Var{a}}{A}\SigmaTimes{\Var{b}}{B}\ldots\Times
\SigmaTimes{\Var{z}}{Z} \Unit\). To reduce the syntactic burden of
these nested pairs, we elaborate a LISP-inspired tuple notation:
\[
\Axiom{\ElabCheck{\Gamma}{()}{\Unit}
                 {\Void}} \qquad
\Rule{\ElabCheck{\Gamma}{\Meta{x}}{\Meta{A}}
                {\Meta{x'}} \quad
      \ElabCheck{\Gamma}{(\Meta{xs})}{\Meta{B}[\Meta{x'}/\Meta{a}]}
                {\Meta{xs'}}}
     {\ElabCheck{\Gamma}{(\Meta{x}\: \Meta{xs})}{\SigmaTimes{\Var{a}}{\Meta{A}}{\Meta{B}}}
                {\PairAnn{\Meta{x'}}{\Meta{xs'}}{a. B}}}
\]

%% **** -> Tuple for EnumU and its elim

\paragraph{Finite sets, introduction and elimination:}
In Section~\ref{sec:type-theory}, we have used an informal set-like
notation for enumerations: we can make that notation formal through
elaboration. To do so, we extend the type checker with the
following rules:
\[\Code{
\Axiom{\ElabCheck{\Gamma}{\Collection{}}{\EnumU}{\NilEnum}} \qquad
\Rule{\ElabCheck{\Gamma}{\Collection{\Meta{ts}}}{\EnumU}{\Meta{E}}}
     {\ElabCheck{\Gamma}{\Collection{\Tag{a}, \Meta{ts}}}{\EnumU}{\ConsEnum[\Tag{a}\: \Meta{E}]}}
\\
\\
\Rule{\ElabCheck{\Gamma}
                {\Collection{\Tag{l_0}, \ldots, \Tag{l_k}}}
                {\EnumU}
                {\Meta{E}} 
       \qquad
       \ElabCheck{\Gamma}
                {(\Meta{e_0} \ldots \Meta{e_k})}
                {\spi[\Meta{E}\: \Meta{P}]}
                {\Meta{es}}}
     {\ElabCheck{\Gamma}
                {\CollectionElim{\Tag{l_0} \mapsto \Meta{e_0}, \ldots \Tag{l_k} \mapsto \Meta{e_k}}}
                {\PiTo{\Var{e}}{\EnumT[\Meta{E}]} \Meta{P}\:\Var{e}}
                {\EnumElim[\Meta{E}\: \Meta{P}\: \Meta{es}]}}
}\]

\paragraph{Indexing finite sets:}
Also, rather than indexing into enumerations through the \(\EnumT\)
codes, we would like to be able to write the label and have it
elaborate to the corresponding index. This is achieved by the
following extension:
\[
\Axiom{\ElabCheck{\Gamma}{\Tag{t}}{\ConsEnum[\Tag{t}\: E]}{\ZeroEnumT}}
\qquad
\Rule{\ElabCheck{\Gamma}{\Tag{t}}{E}{n}}
     {\ElabCheck{\Gamma}{\Tag{t}}{\ConsEnum[\Tag{u}\: E]}{\SucEnumT[n]}}
\]

%% **** -> Datatype constructor

\paragraph{Datatype constructors:}
Finally, while we do not yet have a proper syntax for \emph{declaring}
inductive types, we can already extend our term language with
constructors. Upon elaborating an external term \(\Meta{c}\: a_0
\ldots a_k\) against the fix-point of a tagged description, we replace
this elaboration problem with the one consisting of elaborating the
tuple consisting of the constructor label and the arguments:
\[
\Rule{\ElabCheck{\Gamma}{\In[(\Tag{c}\: a_0 \ldots a_k)]}{\Mu[(\Dsigma[\Meta{E}\: \Meta{T}])]}{t}}
     {\ElabCheck{\Gamma}{\Meta{c}\: a_0 \ldots a_k}{\Mu[(\Dsigma[\Meta{E}\: \Meta{T}])]}{t}}
\qquad
\Rule{\ElabCheck{\Gamma}{\In[(\Tag{c}\: a_0 \ldots a_k)]}{\IMu[(\Dsigma[\Meta{E}\: \Meta{T}])\: \Meta{i}]}{t}}
     {\ElabCheck{\Gamma}{\Meta{c}\: a_0 \ldots a_k}{\IMu[(\Dsigma[\Meta{E}\: \Meta{T}])\: \Meta{i}]}{t}}
\]

\subsection{Elaborating Programs}

\newcommand{\Label}{\Canonical{label}}
\newcommand{\LabelDef}[2]{\blue{\langle} #1 \mathop{\blue{:}} #2 \blue{\rangle}}
\Spacedcommand{\LabelRet}{\Constructor{return}}
\newcommand{\LabelCall}[3]{\Function{call} \green{\langle} #1 \mathop{\blue{:}} #2 \green{\rangle} #3}

%% **** <- Motivation: View presentation special case of our system
%% ***** -> Restate it in more general framework
%% ***** -> 

In the previous Section, we have define type-specific languages,
relying on the types to guide the elaboration process. In this
Section, we go a step further and purposely define finely indexed
types for the purpose of elaboration. This idea of using types as a
\emph{presentation} of our high-level intentions rather than a
\emph{representation} of low-level restrictions originated
in~\citet{mcbride.mckinna:view-from-the-left}. To illustrate our
point, we redevelop the elaboration of programs presented in that
paper and adapt it to our framework. We shall see how these
presentation types guide the elaboration of programs and let us regain
a pattern-matching notation without hard-wiring pattern-matching
itself.

%% **** <- Figure

\begin{figure}[tbp]

\centering

\subfloat[][Programming label]{
\(\Code[c]{
\Rule{\ContextValid{\Gamma}}
     {\TypeJudgment{\Gamma}{\mathbf{f}}{\Label}}
\qquad
\Rule{\TypeJudgment{\Gamma}{\Meta{l}}{\Label} \quad
      \TypeJudgment{\Gamma}{\Meta{t}}{\Meta{T}}}
     {\TypeJudgment{\Gamma}{\Meta{l}\: \Meta{t}}{\Label}}
\\
\Rule{\TypeJudgment{\Gamma}{l}{\Label} \qquad
      \TypeJudgment{\Gamma}{T}{\Set}}
     {\TypeJudgment{\Gamma}{\LabelDef{l}{T}}{\Set}}
\qquad
\Rule{\TypeJudgment{\Gamma}{t}{T} \qquad
      \TypeJudgment{\Gamma}{l}{\Label}}
     {\TypeJudgment{\Gamma}{\LabelRet[t]}{\LabelDef{l}{T}}}
\qquad
\Rule{\TypeJudgment{\Gamma}{c}{\LabelDef{l}{T}}}
     {\TypeJudgment{\Gamma}{\LabelCall{l}{T}{c}}{T}}
}\)
}
\\
\subfloat[][Program declaration]{
\(\Code[c]{
  \boxed{\ElabLet{\Gamma}
                 {f}
                 {\overrightarrow{\PiTel{\Var{x}}{X}}}
                 {T}
                 {t}
                 {\Delta}}
\\
\Rule{\Code{
      \ElabCheck{\Gamma}{\overrightarrow{\PiTel{\Var{x}}{X}} \To T}{\Set}
                {\overrightarrow{\PiTel{\Var{x}}{X'}} \To T'} \\
      \ElabProg{\Gamma, \overrightarrow{\PiTel{\Var{x}}{X'}}}
               {t}
               {\LabelDef{\mathbf{f}\: \vec{x}}{T'}}
               {t'}
     }}
     {\ElabLet{\Gamma}
              {f}
              {\overrightarrow{\PiTel{\Var{x}}{X}}}
              {T}
              {t}
              {\Gamma [ f \mapsto \TypeAnn{\LamAnn{\vec{x}}{\vec{X'}}
                                                  {\LabelCall{\mathbf{f}\:\vec{x}}{T'}{t'}}}
                                          {\overrightarrow{\PiTel{\Var{x}}{X'}} \To T'}]}}
}\)
\label{fig:program-decl}
}
\\
\subfloat[][Program definition]{
\(\Code[c]{
  \boxed{\ElabProg{\Gamma}
                  {t}
                  {\LabelDef{l}{T}}
                  {t'}}
\\
\Rule{\ElabRec{l}{pt} \qquad
      \ElabCheck{\Gamma}{e}{T}{e'}}
     {\ElabProg{\Gamma}
               {pt \DoReturn e}
               {\LabelDef{l}{T}}
               {\LabelRet[e']}}
\qquad
\Rule{\ElabRec{l}{pt} 
      \qquad 
      \Code{
      \ElabEWM{\Gamma}
              {e}
              {e'}
              {\overrightarrow{(\overrightarrow{\PiTel{\Var{x_i}}{X_i}} \To \LabelDef{l_i}{S_i})} \To \LabelDef{l}{T}} \\
      \ElabProg{\Gamma, \overrightarrow{\PiTel{x_i}{X_i}}}
               {p_i}
               {\LabelDef{l_i}{S_i}}
               {s_i}
     }}
     {\ElabProg{\Gamma}
               {pt \DoBy e \{ \vec{p_i} \}}
               {\LabelDef{l}{T}}
               {e'\: (\LamAnn{\vec{x_i}}{\vec{X_i}}{s_i})}}
}\)
\label{fig:program-def}}

\caption{Programming labels}
\label{fig:prog-label}

\end{figure}

%% **** <- Label Prog
%% ***** <- Reprensented by programming task
%% ****** <- Phantom type around type
%% ****** -> Represent decomposition of arguments

%% \subsection{Programming labels}

To guide elaboration of programs, we define \emph{programming label}
types~(Fig.~\ref{fig:prog-label}). In essence, a programming label
\(\LabelDef{l}{T}\) is a phantom type around the type \(T\). However,
the label \(l\) is crucial in that it represents the arguments of the
function we are implementing. A motivating example of label types is
the definition of addition by explicitly using the elimination
principle of natural numbers:
\[
\Let{\PiTel{\Var{m}}{\Nat} & \NatPlus & \PiTel{\Var{n}}{\Nat}}
    {\LabelDef{\Tag{+}\: \Var{m}\: \Var{n}}{\Nat}}{
\Return{\Var{m} & \NatPlus & \Var{n}}
       {\NatInd
           \begin{array}[t]{l}
             (\Lam{\Var{m'}}{\LabelDef{\Tag{+}\: \Var{m'}\: \Var{n}}{\Nat}})\: \Var{m}\\
             \HoleAnn{\LabelDef{\Tag{+}\: \Zero\: \Var{n}}{\Nat}}\\
             \HoleAnn{\PiTo{\Var{m}}{\Nat} 
                      \LabelDef{\Tag{+}\: \Var{m}\: \Var{n}}{\Nat} \To
                      \LabelDef{\Tag{+}\: (\Suc[\Var{m}])\: \Var{n}}{\Nat}}
           \end{array}}}
\]
Thanks to the label types, we can relate every case of the elimination
principle with the state of our programming problem.

%% **** <- Elaboration

%% ***** <- Elaboration of 'let'
%% ***** <- Elaboration of programs
%% ***** <- Elaboration of expression call
%% ***** <- Elimination with a motive
%% ****** -> Assumed
%% ******* -> cf. EWM, EDPM, Equation

Using these label types, we can extend our external language to ease
programming. First, we introduce a programming problem by the
\(\Define\) command~(Fig.~\ref{fig:program-decl}). The elaboration
judgment is subject to the following soundness property:
\begin{theorem}[Soundness of program elaboration]
If \(\ElabLet{\Gamma}
             {f}
             {\overrightarrow{\PiTel{\Var{x}}{X}}}
             {T}
             {t}
             {\Delta}\), 
then \(\ContextValid{\Delta}\).
\end{theorem}

To elaborate the definition of the program, we restrict our attention
to two constructs: returning a value using the \emph{Return}
(\(\DoReturn\)) gadget and appealing to an elimination principle using
the \emph{By} (\(\DoBy\)) gadget. We shall ignore views and other
syntactic conveniences originally introduced
by~\citet{mcbride.mckinna:view-from-the-left}. The elaboration rules
are given in Figure~\ref{fig:program-def}. The idea is that they are
two rules to realize a programming label: either we return a term, or
we decompose the problem further using an eliminator. We have coded
the \emph{programming grammar} within this specification. Remark the
presence of patterns \(pt\) on the left hand side to mimic a
pattern-matching notation. Intuitively, the relation
\(\ElabRec{l}{pt}\) makes sure that the pattern entered by user
corresponds to the label computed by the type-checker. The relation
\(\ElabEWM{\Gamma}{e}{e'}{T}\) abstracts away the task of generating a
term \(e'\) in the internal language from an elimination form \(e\),
as described in \citet{mcbride:elim-2,mcbride:elim}. We will not
recall this construction and simply assume that the generated term is
well-typed with respect to T, \ie \(\TypeJudgment{\Gamma}{e'}{T}\).

%% **** <- Example
%% ***** <- Addition
%% ***** -> [Develop in line with Label Progr and Elaboration spec.]

\if 0

\begin{example}[Elaborating addition]

Rather than dwelling on the abstract specification, let us work
through an example. We shall elaborate the following definition of
addition down to a raw term in our type theory. We assume that natural
numbers and their eliminators \(\NatRec\) and \(\NatCase\) to be
provided: we shall see later that they are actually available already.

In the external language, we write:
%
\[
\Let{\Define\: \PiTel{\Var{m}}{\Nat} & \NatPlus & \PiTel{\Var{n}}{\Nat}}
    {\Nat \Is}{
\By{\Var{m} & \NatPlus & \Var{n}}{\NatRec[\Var{m}]}{
  \By{\quad\Var{m} & \NatPlus & \Var{n}}{\NatCase[\Var{m}]}{
    \Return{\quad\quad\Zero & \NatPlus & \Var{n}}{\Var{n}}
    \Return{\quad\quad(\Suc[\Var{m}]) & \NatPlus & \Var{n}}{\Suc[(\Var{m} \NatPlus \Var{n})]}
  }
}
}
\]
%
%% Working from bottom-up, we elaborate the two return clauses:
%% %
%% \[\Code{
%%   \ElabProg{\TypeAnn{\Var{m}}{\Nat}, \TypeAnn{\Var{n}}{\Nat}}
%%            {\Zero \NatPlus \Var{n} \DoReturn \Var{n}}
%%            {\LabelDef{\Tag{+}\: \Zero\: \Var{n}}{\Nat}}
%%            {\LabelRet[\Var{n}]} \\
%%   \ElabProg{\TypeAnn{\Var{m}}{\Nat}, \TypeAnn{\Var{n}}{\Nat}}
%%            {(\Suc[\Var{m}]) \NatPlus \Var{n} \DoReturn \Suc[(\Var{m} \NatPlus \Var{n})]}
%%            {\LabelDef{\Tag{+}\: (\Suc[\Var{m}])\: \Var{n}}{\Nat}}
%%            {\LabelRet[\Suc[\Var{ih}]]}
%% }\]

\end{example}

\fi

%% * Elaboration of inductive types

\section{Elaborating Datatypes}
\label{sec:elab-data-types}

%% ** <- Motivation
%% *** -> Present elaboration in simple framework
%% *** /> Not using full power of dependent types
%% **** -> Next section
%% *** -> Connect with intuition of datatypes

In this Section, we specify the elaboration of inductive types down to
our \(\Desc\) universe. While this universe only captures strictly
positive types, it is a good exercise to understand the general idea
governing the elaboration of inductive definitions. Besides, because
the syntax is essentially the same, our presentation should be easy to
understand for readers familiar with Coq or Agda.

%% ** <- Intuition
%% *** -> Standard sum-of-product form
%% **** -> Translated to sigma-of-sigma
%% *** <- 1st level structure
%% **** <- Choice of constructor
%% **** <- End of declaration
%% *** <- 2nd level structure
%% **** <- Constructor name
%% **** <- Telescope of arguments
%% **** <- Recursive call
%% ***** -> Matched by label

We adopt the standard sum-of-product high-level notation:
\[
\Data{\Canonical{D}}
     {\overrightarrow{\Param{\Var{p}}{\Meta{P}}}}
     {\Set}{
\Emit{\Canonical{D}}
     {\overrightarrow{\Var{p}}}
     {\Constructor{c}_{\red{0}}\: \overrightarrow{\PiTel{\Var{a_0}}{\Meta{T_0}}}}
\OrEmit{\ldots}
\OrEmit{\Constructor{c}_{\red{k}}\: \overrightarrow{\PiTel{\Var{a_k}}{\Meta{T_k}}}}
}
\]
Where the arguments \(\vec{p}\) are parameters. A \(T_i\) can be
recursive, \ie refers to \(\Canonical{D}\:\vec{p}\). Note that it is
crucial that the parameters are the same in the definition and the
recursive calls.

Our translation to code follows the structure of the definition. The
first level structure consists of the choice of constructors and is
translated to a \(\Dsigma\) code over the finite set of
constructors. The second level structure consists of the
\(\Sigma\)-telescope of arguments: these are translated to
right-nested \(\DSigma\) codes. When parsing arguments, we must make
sure that the recursive arguments are valid and translate them to the
\(\DVar\) code.

%% ** <- Label desc 
%% *** <- Represent data-definition task
%% **** -> Name of data type
%% **** -> List of parameters
%% *** -> Constructors elaborated against it
%% *** -> Recursive calls matched against it

\subsection{Description labels}
\label{sec:desc-labels}

\newcommand{\datatel}[1]{\mathrm{datatel}}
\newcommand{\TypeDatatel}[2]{\Judgment{#1}{#2 \:\datatel{}}}

\newcommand{\IntDatatel}[1]{\green{\llbracket} #1 \green{\rrbracket}_P}

\newcommand{\LabelDesc}[1]{\langle #1 \rangle}
\Spacedcommand{\ReturnDesc}{\Constructor{return}}
\newcommand{\CallDesc}[2]{\Function{call}\: \LabelDesc{#1}\: #2}

To guide the elaboration of inductive definitions, we extend the type
theory with \emph{description
  labels}~(Fig.~\ref{fig:desc-label}). Their role is akin to
programming labels: they structure the elaboration task and are used
to ensure that recursive arguments are correctly elaborated.

A description label \(\LabelDesc{l}\) is a list starting with the name
of the datatype being defined and followed by the parameters of that
datatype. It can be thought as a phantom type around \(\Desc\). We can
introduce such type using \(\ReturnDesc\) that takes the (finite) set
of constructors and their respective code: doing so, we ensure that we
are only accepting tagged descriptions. With \(\CallDesc{l}\), we
eliminate \(\ReturnDesc\) by joining constructors and their codes
in a \(\Dsigma\) code, effectively interpreting the choice of
constructors.

\begin{figure}[bt]

\[
\Code[c]{
\Axiom{\TypeDatatel{\Gamma}{\textbf{D}}} 
\qquad
\Rule{\Code{
      \TypeDatatel{\Gamma}{\Meta{l}} \\
      \TypeJudgment{\Gamma}{\Meta{t}}{\Meta{T}}}}
     {\TypeDatatel{\Gamma}{\Meta{l}\: \Meta{t}}}
\\
\\
\Rule{\TypeDatatel{\Gamma}{l}}
     {\TypeJudgment{\Gamma}{\LabelDesc{l}}{\Set[1]}}
\qquad
\Rule{\Code{
    \TypeJudgment{\Gamma}{E}{\EnumU} \\
    \TypeJudgment{\Gamma}{T}{\EnumT[E] \To \Desc}}}
     {\TypeJudgment{\Gamma}{\ReturnDesc[E\: T]}{\LabelDesc{l}}}
\qquad
\Rule{\TypeJudgment{\Gamma}{t}{\LabelDesc{l}}}
     {\TypeJudgment{\Gamma}{\CallDesc{l}{t}}{\Desc}}
}\]

\caption{Description label}
\label{fig:desc-label}

\end{figure}

%% ** <- Elaborating inductive types

\subsection{Elaborating inductive types}

%% *** <- Adopt top-down approach
%% **** -> 1) See how pieces fit together
%% ***** -> Then, disasemble the pieces
%% **** /> 2) Work backward for proofs
%% ***** -> From elaborated atoms, build term
%% ***** -> Soundness garantees compose together
%% ****** -> Gives soudness of elaboration of 'data'

\Spacedcommand{\Tree}{\Canonical{Tree}}
\Spacedcommand{\FatTree}{\mathbf{Tree}}
\newcommand{\TreeLeaf}{\Constructor{leaf}}
\Spacedcommand{\TreeNode}{\Constructor{node}}

We shall present our translation in a top-down manner: from a complete
definition, we show how the pieces fit together, giving some intuition
for the subsequent translations. We then move on to disassemble and
interpret each sub-component separately. As we progress, the reader
should check that the intuition we gave for the whole is indeed
valid. Every elaboration step is backed by a soundness property:
proving these properties is inherently bottom-up. Only after having
developed all our definition will we be able to prove the soundness of
elaboration of inductive definitions. However, the proof is
technically unsurprising: we shall briefly sketch it at the end of
this Section.  To further ease the understanding of our machinery, we
shall illustrate every step by elaborating the following definition of
binary trees:
\[
\Data{\Tree}
     {\Param{\Var{A}}{\Set}}
     {\Set}{
\Emit{\Tree}
     {\Var{A}}
     {\TreeLeaf}
\OrEmit{\TreeNode[\PiTel{\Var{l}}{\Tree[\Var{A}]}
                  \PiTel{\Var{a}}{\Var{A}}
                  \PiTel{\Var{r}}{\Tree[\Var{A}]}]}
}
\]

%% *** Figure

\begin{figure}[htp]

\centering

%% **** Definition

\subfloat[][Elaboration of definition]{
\(\Code[c]{
  \boxed{\ElabData{\Gamma}
                  {D}
                  {\overrightarrow{\PiTel{\Var{p}}{P}}}
                  {\Set}
                  {\mathrm{choices}}
                  {\Delta}}
\\
\\
\Rule{\Code{
    \ElabCheck{\Gamma}
              {\overrightarrow{\PiTel{\Var{p}}{P}} \To \Set}{\Set[1]}
              {\overrightarrow{\PiTel{\Var{p}}{P'}} \To \Set} \\
    \ElabChoices{\Gamma, \overrightarrow{\TypeAnn{p}{P'}}}
                {\mathrm{choices}}
                {\LabelDesc{\mathbf{D}\: \vec{p}}}
                {\mathrm{code}}}}
     {\ElabData{\Gamma}
               {D}
               {\overrightarrow{\Param{p}{P}}}
               {\Set}
               {\mathrm{choices}}
               {\Gamma [D \mapsto 
                   \TypeAnn{\Lam{\vec{p}}{\Mu[(\CallDesc{\mathbf{D}\: \vec{p}}{\mathrm{code}})]}}
                           {\overrightarrow{\PiTel{\Var{p}}{P'}} \To \Set}]}
               }
}\)
\label{fig:elab-desc-data}
}
%% **** Choices
\\
\subfloat[][Elaboration of choices]{
\(\Code[c]{
\boxed{\ElabChoices{\Gamma}
                      {\mathrm{choices}}
                      {\LabelDesc{l}}
                      {\mathrm{code}}}
\\
\\
\Rule{\ElabConstr{\Gamma}
                 {c_i}
                 {\LabelDesc{l}}
                 {t_i}
                 {\mathrm{code}_i} 
      \qquad
      \Code[c]{
        \ElabCheck{\Gamma}
                  {\Collection{t_i}}
                  {\EnumU}
                  {\Meta{E}} \\
        \ElabCheck{\Gamma}
                  {\CollectionElim{t_i \mapsto \mathrm{code}_i}}
                  {\EnumT[\Meta{E}] \To \Desc}
                  {\Meta{T}}
      }}
     {\ElabChoices{\Gamma}
                  {c_0 | \ldots | c_n}
                  {\LabelDesc{l}}
                  {\ReturnDesc[\Meta{E}\: \Meta{T}]}}
}\)
\label{fig:elab-desc-choices}
}
\\
%% **** Constructor 
\subfloat[][Elaboration of constructor]{
\(\Code[c]{
\boxed{\ElabConstr{\Gamma}
                     {c}
                     {\LabelDesc{l}}
                     {t}
                     {\mathrm{code}}}
\\
\\
\Rule{\Code{
      \ElabCheck{\Gamma}
                {\Tag{t}}
                {\UId}
                {t'} \\
      \ElabArg{\Gamma}
              {\mathrm{args}}
              {\LabelDesc{l}}
              {\mathrm{code}}
     }}
     {\ElabConstr{\Gamma}
                 {t\: \mathrm{args}}
                 {\LabelDesc{l}}
                 {t'}
                 {\mathrm{code}}}
}\)
\label{fig:elab-desc-constr}
}
\\
%% **** Arguments
\subfloat[][Elaboration of arguments]{
\(\Code[c]{
\boxed{\ElabArg{\Gamma}
                 {\mathrm{args}}
                 {\LabelDesc{l}}
                 {\mathrm{code}}}
\\
\\
\Rule{
    \ElabCheck{\Gamma}{T}{\Set}{T'} \qquad
    \ElabArg{\Gamma, \TypeAnn{\Var{x}}{T'}}
            {\Delta}
            {\LabelDesc{l}}
            {\mathrm{code}_\Delta}}
     {\ElabArg{\Gamma}
              {\PiTel{\Var{x}}{T} \Delta}
              {\LabelDesc{l}}
              {\DSigma[T'\: \Lam{\Var{x}}{\mathrm{code}_\Delta}]}}
\qquad
\Rule{
    \ElabRec{T}{l} \qquad
    \ElabArg{\Gamma}
            {\Delta}
            {\LabelDesc{l}}
            {\mathrm{code}_\Delta}}
     {\ElabArg{\Gamma}
              {\PiTel{\Var{x}}{T} \Delta}
              {\LabelDesc{l}}
              {\DVar \DTimes \mathrm{code}_\Delta}}
\\
\Rule{
      \ElabCheck{\Gamma}{T}{\Set}{T'}
      \qquad
      \Code{
      \ElabArg{\Gamma, \TypeAnn{t}{T'}}
              {\nabla}
              {\LabelDesc{l}}
              {\mathrm{code}_\nabla} \\
      \ElabArg{\Gamma}
              {\Delta}
              {\LabelDesc{l}}
              {\mathrm{code}_\Delta}
     }}
     {\ElabArg{\Gamma}
              {\PiTel{\Var{f}}{\PiTo{\Var{t}}{T} \nabla} \Delta}
              {\LabelDesc{l}}
              {(\DPi[T'\: \Lam{\Var{t}}{\mathrm{code}_\nabla}]) \DTimes \mathrm{code}_\Delta}}
\qquad
\Axiom{\ElabArg{\Gamma}
               {\epsilon}
               {\LabelDesc{l}}
               {\DUnit}}
}\)
\label{fig:elab-desc-arg}
}
\\
%% **** Recursive checker
\subfloat[][Recursion matching]{
\(\Code[c]{
\boxed{\ElabRec{T}
               {l}}
\\
\\
\Axiom{\ElabRec{D}
               {\textbf{D}}}
\qquad
\Rule{\ElabRec{T}{l}}
     {\ElabRec{\Meta{T}\: \Meta{p}}
              {\Meta{l}\: \Meta{p}}}
}\)
\label{fig:elab-desc-rec}
}

\caption{Elaboration of inductive types}
\label{fig:elab-desc}

\end{figure}

%% *** <- Elaboration of 'data'
%% **** <- Judgment form
%% **** <- Intuition
%% **** <- Soundness

\paragraph{Elaboration of an inductive definition (Fig.~\ref{fig:elab-desc-data}):}
Elaborating an inductive definition extends the input context with the
datatype definition. To obtain this definition, we first elaborate the
parameters and move onto elaborating the choice of constructors,
introducing a description label in the process.

\begin{example}[Elaborating \(\Tree\)]

Applied to our example, we obtain:
\[
\ElabData{\Gamma}
         {\Tree}
         {\PiTel{\Var{A}}{\Set}}
         {\Set}
         {[\mathrm{choices}]}
         {\Gamma [\Tree \mapsto 
                   \TypeAnn{\Lam{\Var{A}}{\Mu[(\CallDesc{\FatTree[\Var{A}]}{[\mathrm{code}]})]}}
                           {\Var{A} \To \Set}]}
\]
where
\begin{align*}
\mathrm{choices} &\triangleq 
\TreeLeaf \:|\: \TreeNode[\PiTel{\Var{l}}{\Tree[\Var{A}]}
                  \PiTel{\Var{a}}{\Var{A}}
                  \PiTel{\Var{r}}{\Tree[\Var{A}]}]
\\
\mathrm{code} &\triangleq 
\ReturnDesc[\Collection{
                        \begin{array}{l}
                          \Tag{\CN{leaf}},\\ 
                          \Tag{\CN{node}}
                        \end{array}}\:
                    \CollectionElim{
                      \begin{array}{l@{\DoReturn}l}
                        \Tag{\CN{leaf}} & \DUnit, \\
                        \Tag{\CN{node}} & \DVar \DTimes \DSigma[\Var{A}\: \Lam{\_} \DVar \DTimes \DUnit]
                      \end{array}
                    }]
\end{align*}

\end{example}

%% *** <- Elaboration of constructors
%% **** <- Judgment form
%% **** <- Intuition
%% **** <- Soundness
%% **** <- Definition

\paragraph{Elaboration of constructors (Fig.~\ref{fig:elab-desc-choices}):} 
To elaborate the choice of constructors, we elaborate each individual
constructor, hence obtaining their respective label and code. We then
return the finite collection of constructor names and their
corresponding codes. This elaboration step is subject to the soundness
property:
\begin{lemma}\label{lemma:elab-desc-choices}
If \(
\left\{
\begin{array}{l}
  \TypeDatatel{\Gamma}{\Meta{l}} \\
  \ElabChoices{\Gamma}
              {\mathrm{choices}}
              {\LabelDesc{l}}
              {\mathrm{code}}
\end{array}\right.\), 
then \(\TypeJudgment{\Gamma}{\mathrm{code}}{\LabelDesc{l}}\)
\end{lemma}

\begin{example}[Elaborating \(\Tree\)]

Applied to our example, we obtain:
\[
\ElabChoices{\Gamma, \TypeAnn{\Var{A}}{\Set}}
            {[\mathrm{choices}]}
            {\LabelDesc{\FatTree[\Var{A}]}}
            {[\mathrm{code}]}
\]
Where \(\mathrm{choices}\) and \(\mathrm{code}\) have been defined
above.

\end{example}

%% *** <- Elaboration of constructor
%% **** <- Judgment form
%% **** <- Intuition
%% **** <- Soundness
%% **** <- Definition

\paragraph{Elaboration of constructor (Fig.~\ref{fig:elab-desc-constr}):}
The role of this elaboration step is twofold. First, we extract the
constructor name and elaborate it into a label. Second, we elaborate
the arguments of that constructor, hence obtaining a \(\Desc\)
code. We return the pair of the label and the arguments' code, subject
to the following soundness property:
\begin{lemma}\label{lemma:elab-desc-constr}
If 
\(
\left\{
\Code{
\TypeDatatel{\Gamma}{l} \\
\ElabConstr{\Gamma}
           {c}
           {\LabelDesc{l}}
           {t}
           {\mathrm{code}}
}
\right.
\), then
\(
\left\{\Code{
\TypeJudgment{\Gamma}{t}{\UId} \\
\TypeJudgment{\Gamma}{\mathrm{code}}{\Desc}
}\right.
\).

\end{lemma}

\begin{example}[Elaborating \(\Tree\)]

Since our datatype definition has two constructors, there are two
instances of constructor elaboration:
\[\Code{
\ElabConstr{\Gamma, \TypeAnn{A}{\Set}}
           {\TreeLeaf}
           {\LabelDesc{\FatTree[\Var{A}]}}
           {\Tag{\CN{leaf}}}
           {\DUnit} \\
\ElabConstr{\Gamma, \TypeAnn{A}{\Set}}
           {\TreeNode[\PiTel{\Var{l}}{\Tree[\Var{A}]}
                      \PiTel{\Var{a}}{\Var{A}}
                      \PiTel{\Var{r}}{\Tree[\Var{A}]}]}
           {\LabelDesc{\FatTree[\Var{A}]}}
           {\Tag{\CN{node}}}
           {\DVar \DTimes \DSigma[\Var{A}\: \Lam{\_} \DVar \DTimes \DUnit]}
}\]

\end{example}

%% *** <- Elaboration of arguments
%% **** <- Judgment form
%% **** <- Intuition
%% **** <- Soundness
%% **** <- Definition

\paragraph{Elaboration of arguments (Fig.~\ref{fig:elab-desc-arg}):}
The last and perhaps most interesting rule is the elaboration of
arguments. Intuitively, the arguments form a telescope of
\(\Sigma\)-types, hence our translation to \(\DSigma\) and \(\DTimes\)
codes. The first two rules are non-deterministic: \(T\) could either
be a proper type or a recursive call. In the first case, this maps to
a standard \(\DSigma\) code, while in the second case, we must make
sure that the recursive call is valid and, if so, we generate a
\(\DVar\) code. We also support exponentials in the definitions,
mapping them to the \(\DPi\) code. Once all arguments have been
processed, we conclude by generating the \(\DUnit\) code. This
translation is subject to the following soundness property:
\begin{lemma}\label{lemma:elab-desc-arg}
If 
\(
\left\{
\begin{array}{l}
\TypeDatatel{\Gamma}{l} \\
\ElabArg{\Gamma}
          {\mathrm{args}}
          {\LabelDesc{l}}
          {\mathrm{code}}
\end{array}
\right.\), then
\(\TypeJudgment{\Gamma}{\mathrm{code}}{\Desc}\).
\end{lemma}

\begin{example}[Elaborating \(\Tree\)]

Elaborating the arguments of the \(\TreeLeaf\) constructor is trivial:
\[
\Axiom{\ElabArg{\Gamma, \TypeAnn{A}{\Set}}
               {\epsilon}
               {\LabelDesc{\FatTree[\Var{A}]}}
               {\DUnit}}
\]
As for the \(\TreeNode\) constructor, we obtain its code through the
following sequence of elaborations:
\[
\Rule{
  \ElabRec{\Tree[\Var{A}]}{\FatTree[\Var{A}]} \qquad
  \Rule{
    \Rule{
      \ElabRec{\Tree[\Var{A}]}{\FatTree[\Var{A}]} \qquad
      \Axiom{\ElabArg{\Gamma, \TypeAnn{A}{\Set}, \TypeAnn{a}{A}}
                     {\epsilon}
                     {\LabelDesc{\FatTree[\Var{A}]}}
                     {\DUnit}}}
         {\ElabArg{\Gamma, \TypeAnn{A}{\Set}, \TypeAnn{a}{A}}
                  {\PiTel{\Var{r}}{\Tree[\Var{A}]}}
                  {\LabelDesc{\FatTree[\Var{A}]}}
                  {\DVar \DTimes \DUnit}}}
       {\ElabArg{\Gamma, \TypeAnn{A}{\Set}}
                {\PiTel{\Var{a}}{\Var{A}}
                 \PiTel{\Var{r}}{\Tree[\Var{A}]}}
                {\LabelDesc{\FatTree[\Var{A}]}}
                {\DSigma[\Var{A}\: \Lam{\_} \DVar \DTimes \DUnit]}}}
     {\ElabArg{\Gamma, \TypeAnn{A}{\Set}}
              {\PiTel{\Var{l}}{\Tree[\Var{A}]}
                \PiTel{\Var{a}}{\Var{A}}
                \PiTel{\Var{r}}{\Tree[\Var{A}]}}
              {\LabelDesc{\FatTree[\Var{A}]}}
              {\DVar \DTimes \DSigma[\Var{A}\: \Lam{\_} \DVar] \DTimes \DUnit}}
\]

\end{example}

%% ** <- Example [in the Spec]
%% *** <- Elaboration of tree
%% ** <- Soundness
%% *** <- Lemma: elaboration of arguments is sound
%% *** <- Lemma: elaboration of constructor is sound
%% *** <- Lemma: elaboration of constructors is sound
%% *** <- Theorem: elaboration of data is sound

%\subsection{Soundness and correctness of elaboration}

We can now prove the soundness of the whole translation. We
formulate soundness as follow:
\begin{theorem}[Soundness of elaboration]

If \(\ElabData{\Gamma}
              {D}
              {\overrightarrow{\PiTel{\Var{p}}{P}}}
              {\Set}
              {\mathrm{choices}}
              {\Delta}\), 
then \(\ContextValid{\Delta}\).

\end{theorem}
\begin{proof}

First, we prove Lemma~\ref{lemma:elab-desc-arg} by induction on the
list of arguments. We then obtain
Lemma~\ref{lemma:elab-desc-constr}. By applying this lemma to all
constructors, we obtain Lemma~\ref{lemma:elab-desc-choices}. From this
Lemma, we finally obtain the soundness theorem.

\end{proof}

%% ** <- Discussion
%% *** -> Translation to Coq declaration
%% *** -> Correctness: induction is equivalent
%% **** <- Easy: Our induction is provable in Coq
%% ***** <- Cf. Levitation
%% **** <- Conjecture: Coq definitions are derivable in our system (with K)
%% ***** <- Cf. Conor thesis, etc.
%% ***** <- Cf. Gimenez

While our soundness theorem gives some hint as to the correctness of
our specification, we could obtain a stronger result by proving an
equivalence between Coq's \texttt{Inductive} definitions and the
corresponding datatype declaration in our system. This equivalence
amounts to proving the equivalence of the associated elimination
forms, \ie \texttt{Fix} in Coq and \(\induction\) in our
system. However, since we do not know of any formal description of
elimination principles generated from an \texttt{Inductive}
definition, we shall use the simpler presentation given in
\citet{gimenez:coq-induction}.

\newcommand{\InterpretConstForm}[1]{\lfloor #1 \rfloor}

The fact that our induction principle is provable in Coq is a known
result~\citep{dagand:levitation}. The other direction consists in
proving that any Coq \texttt{Fix} definition can be implemented using
our induction principle. To this end, we use Gim\'{e}nez reduction of
\texttt{Fix} definitions down to elimination rules. To prove this
result, we first translate Gim\'{e}nez constructor forms to our
universe of code:
\begin{align*}
  \InterpretConstForm{X \vec{N}} 
      &\mapsto
          \DUnit \\
  \InterpretConstForm{(\TypeAnn{x}{M}) \to \mathcal{C}}
      &\mapsto
          \DSigma[M\: \Lam{\Var{x}}{\InterpretConstForm{\mathcal{C}}}] \\
  \InterpretConstForm{(\TypeAnn{x}{M} \to X\: \vec{N}) \to \mathcal{C}}
      &\mapsto
          (\DPi[M\: \Lam{\_}{\DVar}]) \DTimes \InterpretConstForm{\mathcal{C}}
\end{align*}
We thus get a translation from his recursive types declaration to a code in our
universe:
\[
\InterpretConstForm{\mathbf{Ind}(\TypeAnn{x}{A})\langle C_0 \ldots C_n \rangle}
    \mapsto 
    \Dsigma[n\: \CollectionElim{i \mapsto \InterpretConstForm{C_i}}]
\]
Having done that, it is then a straightforward symbol-pushing exercise
to prove that Coq's elimination rules (Section~3.1.1,
\citep{gimenez:coq-induction}) can be reduced to our generic
elimination principle. The crux of the matter consists in showing that
the minor premises -- defined by \(\mathcal{E}_1\) in that paper --
maps to the induction hypothesis described by
\((\PiTo{\Var{d}}{\InterpretDesc{\Var{D}}[(\Mu{\Var{D}})]}
\All[\Var{D}\: (\Mu{\Var{D}})\: \Var{P}\: \Var{d}] \To \Var{P}
(\In{\Var{d}}))\) in our system.

%\todo{Shall we?}

%% * Elaboration of inductive families

\section{Elaborating inductive families}
\label{sec:elab-data-family}

%% ** <- Motivation
%% *** <- Subsumes treatment of inductive types
%% **** <- Add support for indices
%% **** <- Add computation at the index level
%% *** -> Only necessary implementation
%% **** -> Notation collapse
%% *** /> Builds upon previous section

In this Section, we extend our treatment of inductive definitions to
inductive families. To do so, we add support for indices and
computation on these indices. The resulting system subsumes the one
presented in the previous Section. Hence, we shall reuse many
notations used previously: this should not affect reasoning or
implementing this elaborator, since only that last system is
necessary. We shall however develop our intuition of this more
powerful presentation from the simpler translation of inductive types.

%% ** <- Superset of Coq/Agda presentation
%% *** <- Add computation on indices
%% **** -> Disabled by discarding EWM pattern (fig.?)
%% *** -> Alternative to dotdotdot hell
%% **** <- Rely on elaboration to reduce to codes
%% **** -> Work on code/semantics

The syntax we elaborate is strongly inspired by the syntax used by
Agda and Coq. However, to support computation over indices, we support
the Epigram-style \emph{by} gadget (\(\DoBy\)). Our language of
inductive definition is therefore more complex, following the
skeleton:
\[
\Data{\Canonical{D}}
     {\overrightarrow{\Param{\Var{p}}{P}} 
      \overrightarrow{\Index{\Var{i}}{I}}}
     {\Set}{
\Emit{\Canonical{D}}
     {\vec{\Var{p}}\: 
      \overrightarrow{\Constraint{\Var{i}}{t_0}}}
     {c_{0,0} \overrightarrow{\PiTel{\Var{a_{0,0}}}{T_{0,0}}}}
\OrEmit{\ldots}
%\OrEmit{c_{0,l} \overrightarrow{\PiTel{\Var{a_{0,l}}}{T_{0,l}}}}
&\multicolumn{5}{@{}l}{\vdots} \\
%\Emit{\Canonical{D}}
%     {\vec{\Var{p}}\: 
%      \overrightarrow{\Constraint{\Var{i}}{t_k}}}
%     {\ldots}
\By{\Canonical{D} & 
    \vec{\Var{p}}\: 
    \vec{\Var{i}}}
   {\Function{elim}\: \Var{i_m}}{
\Emit{\quad\Canonical{D}}
     {\vec{\Var{p}}\: 
      \vec{\Var{k_0}}}
     {\ldots}
&\multicolumn{5}{@{}l}{\quad\vdots} \\
%% \Emit{\quad\Canonical{D}}
%%      {\vec{\Var{p}}\: 
%%       \vec{\Var{k_0}}}
%%      {c_{k+1,0} \overrightarrow{\PiTel{\Var{a_{k+1,0}}}{T_{k+1,0}}}}
%\OrEmit{\ldots}
%\OrEmit{c_{k+1,j} \overrightarrow{\PiTel{\Var{a_{k+1,j}}}{T_{k+1,j}}}}
%&\multicolumn{5}{@{}l}{\quad\vdots} \\
%\Emit{\quad\Canonical{D}}
%     {\vec{\Var{p}}\: 
%       \vec{\Var{k_p}}}
%     {\ldots}
}}
\]

For anyone wanting to reason about Agda or Coq definitions, it is
straightforward to simply discard the elaboration of computation over
indices. Thus, with minor adjustement, our formalization of
elaboration gives a translation semantics to inductive definitions
used in these main-stream theorem provers.

%% ** <- Label IDesc

\subsection{Description labels}

\newcommand{\idatatel}[1]{\mathrm{idatatel}_{#1}}
\newcommand{\TypeIDataTel}[2]{\Judgment{#1}{#2 \:\idatatel{}}}
\newcommand{\InterpretIDataTel}[1]{\green{\llbracket} #1 \green{\rrbracket_D}}

%% *** <- datatel not enough
%% **** <- Only cope with Parameters
%% **** -> Need support for indices
%% ***** <- Either unconstrained
%% ***** <- Or constrained to some value
%% **** -> Need to export index types
%% ***** <- Resulting desc indexed by these
%% ***** -> Collect index types through the label

In Section~\ref{sec:desc-labels}, we have introduced the notion of
description labels. While it was enough to deal with
parameterized definitions, we have to extend this notion to
account for indexing. Besides, indexing can be either unconstrained --
some index \(\IndexBracket{i}\) -- or it can be constrained to some
particular value -- such as \(\Constraint{i}{t}\). Besides, a label is
not a phantom type around a description but around an \(\IDesc[I]\)
where \(I\) corresponds to the product of all index types the label is
taking as arguments. 

%% *** <- Extend label accordingly
%% ***** <- Index desc by indices from idatatel
%% ***** -> Small change

These requirements lead us to the definition of description labels
given in Figure~\ref{fig:idesc-label}. We define label's arguments
through a telescope extended with indices and constraints. To compute
the index type of the resulting description, we introduce the
\(\InterpretIDataTel{\_}\) function. Accordingly, we update the introduction
and elimination rule of labels to account for the index type computed
from the telescope.

\begin{figure}
\[\Code[c]{
\Code[llll]{
\Rule{\ContextValid{\Gamma}}
     {\TypeIDataTel{\Gamma}{\mathbf{D}}}
&
\Rule{\Code{
      \TypeIDataTel{\Gamma}{\Meta{l}} \\
      \TypeJudgment{\Gamma}{\Meta{p}}{\Meta{P}}}}
     {\TypeIDataTel{\Gamma}{\Meta{l}\: \ParamBracket{\Meta{p}}}}
&
\Rule{\Code{
      \TypeIDataTel{\Gamma}{\Meta{l}} \\
      \TypeJudgment{\Gamma}{\Meta{i}}{\Meta{I}}}}
     {\TypeIDataTel{\Gamma}{\Meta{l}\: \IndexBracket{\Meta{i}}}}
&
\Rule{\Code{
      \TypeIDataTel{\Gamma}{\Meta{l}} \\
      \TypeAnn{\Meta{i}}{\Meta{I} \in \Gamma}}}
     {\TypeIDataTel{\Gamma}{\Meta{l}\: \Constraint{\Meta{i}}{\Meta{t}}}}
\\
\\
\InterpretIDataTel{\mathbf{D}} \mapsto \Unit
&
\InterpretIDataTel{\Var{l}\: \ParamBracket{\Var{p}}} \mapsto \InterpretIDataTel{\Var{l}}
&
\InterpretIDataTel{\Var{l}\: \IndexBracket{\Var{i}}} \mapsto \InterpretIDataTel{\Var{l}} \Times \Var{I}
&
\InterpretIDataTel{\Var{l}\: \Constraint{\Var{i}}{\Var{t}}} \mapsto \InterpretIDataTel{\Var{l}} \Times \Var{I}
}
\\
\\
\Rule{\TypeDatatel{\Gamma}{l}}
     {\TypeJudgment{\Gamma}{\LabelDesc{l}}{\Set[1]}}
\qquad
\Rule{\Code{
    \TypeJudgment{\Gamma}{E}{\EnumU} \\
    \TypeJudgment{\Gamma}{T}{\EnumT[E] \To \IDesc[\InterpretIDataTel{\Meta{l}}]}}}
     {\TypeJudgment{\Gamma}{\ReturnDesc[E\: T]}{\LabelDesc{l}}}
\qquad
\Rule{\TypeJudgment{\Gamma}{t}{\LabelDesc{l}}}
     {\TypeJudgment{\Gamma}{\CallDesc{l}{t}}{\IDesc[\InterpretIDataTel{\Meta{l}}]}}
}
\]

\caption{Description label (indexed)}
\label{fig:idesc-label}

\end{figure}

%% *** <- [IGNORED] Need syntax for pattern
%% **** <- Desc: abused function application to represent patterns
%% ***** <- Parameters passed as normal arguments
%% ****** -> Relation \ElabRec{T}{l} (fig.\ref{fig:elab-desc-rec})
%% ****** -> Could be extended to index
%% ***** /> Does not deals with constrained index
%% ****** <- Used on LHS of data-declarations
%% ****** -> Define data pattern
%% ******* -> Elaborated against datatel

\newcommand{\IDataPattern}{\mathrm{patt}}
\newcommand{\TypeIDataPattern}[2]{\Judgment{#1}{#2 \:\IDataPattern}}

\if 0

Unfortunately, this is not enough. Since we want to be explicit about
constraints on indices, we must introduce another syntactic artifact
to represent these constraints.

\[
\Rule{\ContextValid{\Gamma}}
     {\TypeIDataPattern{\Gamma}{\mathbf{D}}}
\qquad
\Rule{\TypeIDataPattern{\Gamma}{\Meta{pt}} \quad
      \TypeJudgment{\Gamma}{\Meta{t}}{\Meta{T}}}
     {\TypeIDataPattern{\Gamma}{\Meta{pt}\: \Meta{t}}}
\qquad
\Rule{\TypeIDataPattern{\Gamma}{\Meta{pt}} \quad
      \TypeAnn{\Meta{i}}{\Meta{I}} \in \Gamma \quad
      \TypeJudgment{\Gamma}{\Meta{t}}{\Meta{I}}}
     {\TypeIDataPattern{\Gamma}{\Meta{pt}\: \Constraint{\Meta{i}}{\Meta{t}}}}
\]

\fi

%% ** <- Figure

\renewcommand{\ElabChoices}[5]
           {#1 \vdash #3 \ni #2 
                 \ElabRel{Cs}
                     [#4 \mapsto #5]}

\begin{figure}[tbp]

\centering

%% *** Elaboration of definition

\subfloat[][Elaboration of definition]{
\(\Code[c]{
\boxed{\ElabIData{\Gamma}
                  {D}
                  {\overrightarrow{\Param{\Var{p}}{P}}}
                  {\overrightarrow{\Index{\Var{i}}{I}}}
                  {\Set}
                  {\mathrm{choices}}
                  {\Delta}}
\\
\\
\Rule{\Code{
    \ElabCheck{\Gamma}
              {\overrightarrow{\PiTel{\Var{p}}{P}}
               \overrightarrow{\PiTel{\Var{i}}{I}} \To \Set}{\Set[1]}
              {\overrightarrow{\PiTel{\Var{p}}{P'}}
               \overrightarrow{\PiTel{\Var{i}}{I'}} \To \Set} \\
    \ElabDataPatts{\Gamma; 
                   \overrightarrow{\TypeAnn{p}{P'}};
                   \overrightarrow{\TypeAnn{i}{I'}}}
                  {\mathrm{choices}}
                  {\LabelDesc{\mathbf{D}\: 
                              \overrightarrow{\ParamBracket{p}}\: 
                              \overrightarrow{\IndexBracket{i}}}}
                  {\mathrm{code}}}}
     {\Code{
         \ElabIData{\Gamma}
                   {D}
                   {\overrightarrow{\Param{p}{P}}}
                   {\overrightarrow{\Index{i}{I}}}
                   {\Set}
                   {\mathrm{choices}}
                   {\\\qquad
                     \Gamma [D \mapsto 
                       \TypeAnn{\LamAnn{\vec{p}}{\vec{P'}}
                                    {\Mu[(\LamAnn{\vec{i}}{\vec{I'}}
                                          \CallDesc{D\: \overrightarrow{\ParamBracket{p}}\:
                                                        \overrightarrow{\IndexBracket{i}}}
                                                   {\mathrm{code}})]}}
                               {\overrightarrow{\PiTel{\Var{p}}{P'}}
                                \overrightarrow{\PiTel{\Var{i}}{I'}} \To \Set}]}}}
}\)
\label{fig:idesc-elab-idata}
}
\\
%% *** Elaboration of patterns
\subfloat[][Elaboration of patterns]{
\(\Code[c]{
\boxed{\ElabDataPatts{\Gamma}
                     {\mathrm{patts}}
                     {\LabelDesc{l}}
                     {\mathrm{code}}}
\\
\\
\Rule{
    \Code{
    \ElabIndices{l}{pt_i}{l_i} \\
    \ElabChoices{\Gamma}
                {cs_i}
                {\LabelDesc{l_i}}
                {\Collection{\vec{c_{i,j}}}}
                {\CollectionElim{\overrightarrow{c_{i,j} \mapsto as_{i,j}}}}}
    \qquad
    \Code{
    \ElabIndices{l}{pt_{i+1}}{l_{i+1}} \\
    \ElabEWM{\Gamma}{e}{e'}
            {\overrightarrow{(\overrightarrow{\PiTel{\Var{x_k}}{X_k}} \To \LabelDesc{l_k})} \To \LabelDesc{l_{i+1}}} \\
    \ElabDataPatts{\Gamma, \overrightarrow{\TypeAnn{x_k}{X_k}}}
                  {p_k}
                  {\LabelDesc{l_k}}
                  {\mathrm{code}_k}}}
     {\ElabDataPatts{\Gamma}
                    {\Code{
                        pt_0 \backepsilon cs_0 \\
                        \vdots \\
                        pt_i \backepsilon cs_i \\
                        pt_{i+1} \DoBy e \:\{ \vec{p_k} \}}}
                    {\LabelDesc{l}}
                    {\ReturnDesc[\Collection{\vec{c_{i,j}}, \Tag{\mathrm{elim}}}\:
                                 \CollectionElim{c_{i,j} \mapsto as_{i,j}, 
                                                \Tag{\mathrm{elim}} \mapsto e'\: \overrightarrow{(\LamAnn{\vec{x_k}}{\vec{X_k}}{\CallDesc{l}{\mathrm{code}_k}})}}]}}
}\)
\label{fig:idesc-elab-datapatts}
}
\\
%% *** Elaboration of choices
\subfloat[][Elaboration of choices]{
\(\Code[c]{
\boxed{\ElabChoices{\Gamma}
                   {\mathrm{choices}}
                   {\LabelDesc{l}}
                   {E}
                   {T}}
\\
\\
\Rule{\ElabConstr{\Gamma}
                 {c_i}
                 {\LabelDesc{l}}
                 {t_i}
                 {\mathrm{code}_i}}
     {\ElabChoices{\Gamma}
                  {c_0 | \ldots | c_n}
                  {\LabelDesc{l}}
                  {\Collection{\vec{t_i}}}
                  {\CollectionElim{\overrightarrow{t_i \mapsto \mathrm{code}_i}}}}
}\)
\label{fig:idesc-elab-choices}
}
\\
%% *** Elaboration of constructor
\subfloat[][Elaboration of constructor]{
\(\Code[c]{
\boxed{\ElabConstr{\Gamma}
                  {c}
                  {\LabelDesc{l}}
                  {t}
                  {\mathrm{code}}}
\\
\\
\Rule{\ElabCheck{\Gamma}
                {\Tag{c}}
                {\UId}
                {c'}
      \qquad
      \ElabArg{\Gamma}
              {\mathrm{args}}
              {\LabelDesc{l}}
              {\mathrm{code}}}
     {\ElabConstr{\Gamma}
                 {c\: \mathrm{args}}
                 {\LabelDesc{l}}
                 {c'}
                 {\mathrm{code}}}
}\)
\label{fig:idesc-elab-constr}
}

\caption{Elaboration of inductive families (1)}
\label{fig:elab-families-1}

\end{figure}

\begin{figure}

\centering
%% *** Elaboration of arguments
\subfloat[][Elaboration of arguments]{
\(\Code[c]{
\boxed{
  \ElabArg{\Gamma}
          {\mathrm{args}}
          {\LabelDesc{l}}
          {\mathrm{code}}
}
\\
\\
\Rule{\ElabCheck{\Gamma}{T}{\Set}{T'} \qquad
      \ElabArg{\Gamma, \TypeAnn{x}{T'}}
              {\Delta}
              {\LabelDesc{l}}
              {\mathrm{code}_\Delta}}
     {\ElabArg{\Gamma}
              {\PiTel{\Var{x}}{T}\: \Delta}
              {\LabelDesc{l}}
              {\DSigma[T\: \Lam{\Var{x}}{\mathrm{code}_\Delta}]}} 
\qquad
\Rule{\ElabIndices{T}{l}{is} \qquad
      \ElabArg{\Gamma}
              {\Delta}
              {\LabelDesc{l}}
              {\mathrm{code}_\Delta}}
     {\ElabArg{\Gamma}
              {\PiTel{\Var{x}}{T} \Delta}
              {\LabelDesc{l}}
              {(\DVar[is]) \DTimes \mathrm{code}_\Delta}}
\\
\Rule{\ElabCheck{\Gamma}{T}{\Set}{T'} \qquad
      \Code{
        \ElabArg{\Gamma, \TypeAnn{x}{T'}}
                {\nabla}
                {\LabelDesc{l}}
                {\mathrm{code}_\nabla} \\
        \ElabArg{\Gamma}
                {\Delta}
                {\LabelDesc{l}}
                {\mathrm{code}_\Delta}}}
     {\ElabArg{\Gamma}
              {\PiTel{\Var{f}}{\PiTo{\Var{x}}{T} \nabla} \Delta}
              {\LabelDesc{l}}
              {(\DPi[T'\: \Lam{\Var{x}}{\mathrm{code}_\nabla}]) \DTimes \mathrm{code}_\Delta}}
\qquad
\Rule{\ElabEqs{\Gamma}
              {l}
              {q}}
     {\ElabArg{\Gamma}
              {\epsilon}
              {\LabelDesc{l}}
              {q}}
}\)
\label{fig:idesc-elab-arg}
}
\\
%% *** Pattern validation
\subfloat[][Pattern validation]{
\(\Code[c]{
\boxed{
  \ElabIndices{l}
              {T}
              {l_T}
}
\\
\\
\Axiom{\ElabIndices{\mathbf{D}}
                   {D}
                   {\mathbf{D}}}
\qquad
\Rule{\ElabIndices{l}
                  {T}
                  {l_T}}
     {\ElabIndices{l\: \ParamBracket{p}}
                  {T\: p}
                  {l_T\: \ParamBracket{p}}}
\qquad
\Rule{\ElabIndices{l}
                  {T}
                  {l_T}}
     {\ElabIndices{l\: \IndexBracket{i}}
                  {T\: i}
                  {l_T\: \IndexBracket{i}}}
\\
\Rule{\ElabIndices{l}
                  {T}
                  {l_T}}
     {\ElabIndices{l\: \IndexBracket{i}}
                  {T\: \Constraint{i}{t}}
                  {l_T\: \Constraint{i}{t}}}
\qquad
\Rule{\ElabIndices{l}
                  {T}
                  {l_T}}
     {\ElabIndices{l\: \Constraint{i}{t}}
                  {T\: \Constraint{i}{t}}
                  {l_T\: \Constraint{i}{t}}}
}
\)
\label{fig:idesc-elab-indices}
}
\\
%% *** Elaboration of constraints
\subfloat[][Elaboration of constraints]{
\(\Code[c]{
\boxed{
  \ElabEqs{\Gamma}
          {l}
          {q}
}
\\
\\
\Rule{\ContextValid{\Gamma}}
     {\ElabEqs{\Gamma}
              {\mathbf{D}}
              {\DUnit}}
\qquad
\Rule{\ElabEqs{\Gamma}
              {l}
              {q}}
     {\ElabEqs{\Gamma}
              {l\: \ParamBracket{p}}
              {q}}
\qquad
\Rule{\ElabEqs{\Gamma}
              {l}
              {q}}
     {\ElabEqs{\Gamma}
              {l\: \IndexBracket{i}}
              {q}} 
\qquad
\Rule{\ElabEqs{\Gamma}
              {l}
              {q} 
      \qquad
      \Code{
        \TypeAnn{i}{I} \in \Gamma \\
        \ElabCheck{\Gamma}{t}{I}{t'}}}
     {\ElabEqs{\Gamma}
              {l\: \Constraint{i}{t}}
              {\DSigma[(i \PropEqual t')\: \Lam{\_}{q}]}}
}\)
\label{fig:idesc-elab-eqs}
}
\\
%% *** Extraction of indices
\subfloat[][Extraction of indices]{
\(\Code[c]{
\boxed{
  \ElabRecArgs{\Gamma}{T}{l}{is}
}
\\
\\
\Rule{\ContextValid{\Gamma}}
     {\ElabRecArgs{\Gamma}
                  {D}
                  {\mathbf{D}}
                  {\Void}}
\qquad
\Rule{\ElabRecArgs{\Gamma}
                  {T}
                  {l}
                  {is}}
     {\ElabRecArgs{\Gamma}
                  {T\: p}
                  {l\: \ParamBracket{p}}
                  {is}}
\qquad
\Rule{\ElabRecArgs{\Gamma}
                  {T}
                  {l}
                  {is}}
     {\ElabRecArgs{\Gamma}
                  {T\: i}
                  {l\: \IndexBracket{i'}}
                  {\Pair{is}{i}}}
\qquad
\Rule{\ElabRecArgs{\Gamma}
                  {T}
                  {l}
                  {is}}
     {\ElabRecArgs{\Gamma}
                  {T\: i}
                  {l\: \Constraint{i'}{t}}
                  {\Pair{is}{i}}}
}\)
\label{fig:idesc-elab-recargs}
}

\caption{Elaboration of inductive families (2)}
\label{fig:elab-families-2}

\end{figure}

%% ** <- Methodology
%% *** <- Again, top-down
%% *** <- Architecture similar to inductive types
%% **** <- Elaboration of choices (Fig.?)
%% **** <- Elaboration of constructor (Fig.?)
%% **** <- Elaboration of arguments (Fig.?)
%% *** /> New steps too
%% **** <- Treatment of index computation
%% ***** <- Elaboration of choice/pattern (Fig.?)
%% **** <- Index management
%% ***** <- Elaboration of equality constraints (Fig.?)
%% ***** <- Extraction of indices (Fig.?)

As for inductive types, we shall present the elaboration process in a
top-down manner. This presentation shares some strong commonality with
the simpler elaboration of inductive types: we elaborate choices of
constructors (Fig.~\ref{fig:idesc-elab-choices}), followed by
individual constructors (Fig.~\ref{fig:idesc-elab-constr}), and
finally processing the telescope of arguments
(Fig.~\ref{fig:idesc-elab-arg}). However, the presence of indices
introduces some new steps too. We support constraint of indices and
computation over them through a new top-level rule
(Fig.~\ref{fig:idesc-elab-datapatts}). Besides, we must translate the
constraints to actual equalities (Fig.~\ref{fig:idesc-elab-eqs}) and
pass the correct indices when elaborating a recursive call
(Fig.~\ref{fig:idesc-elab-recargs}).

%% ** <- Elaborating inductive families

\subsection{Elaborating inductive families}

%% *** <- Developing examples along
%% **** <- Constrained Vector
%% **** <- Computational Vector

To give a better intuition of a perhaps intricate system, we should
illustrate every inference rule with two examples. Our first example
is the definition of vectors relying on constraints to enforce the
indexing discipline:
\[
\VectorDefEquality
\]
While our second example consists of the alternative definition of
vector, where we compute over the index to determine which constructor
to offer:
\[
\VectorDef
\]

\Spacedcommand{\FatVector}{\mathbf{Vec}}

%% *** <- Elaboration of 'data'
%% **** <- Judgment form
%% **** <- Intuition
%% **** <- Soundness
%% **** <- Definition

\newcommand{\choicesVecEq}{\mathrm{choices}_{=}}
\newcommand{\codeVecEq}{\mathrm{code}_{=}}

\paragraph{Elaboration of inductive families (Fig.~\ref{fig:idesc-elab-idata}):} 
The elaboration of an inductive definitions sets up the environment to
trigger the elaboration of the choices of constructors. To do so, we
first elaborate the telescope of parameters and indices types. We can
then translate the choices by elaborating against the label type
corresponding to the given inductive type.

\begin{example}[Vector, constrained]
The elaboration of constraint-based vectors starts as follow:
\[\Code{
\Rule{\ElabDataPatts{\TypeAnn{\Var{A}}{\Set},
                     \TypeAnn{\Var{n}}{\Nat}}
                    {[\choicesVecEq]}
                    {\FatVector[\ParamBracket{\Var{A}}\IndexBracket{\Var{n}}]}
                    {[\codeVecEq]}}
     {\Code{
\ElabIData{}
          {\Vector{}}
          {\Param{\Var{A}}{\Set}}
          {\Index{\Var{n}}{\Nat}}
          {\Set}
          {[\choicesVecEq]}
          {\\ \qquad
            \Gamma [\Vector{} \mapsto 
                       \TypeAnn{\LamAnn{\Var{A}}{\Set}
                                   \Mu[(\LamAnn{\Var{n}}{\Nat}
                                          \CallDesc{\FatVector[\ParamBracket{\Var{A}}\: 
                                                               \IndexBracket{\Var{n}}]}
                                                   {[\codeVecEq]})]}
                               {\PiTel{\Var{A}}{\Set}
                                \PiTo{\Var{n}}{\Nat} \Set}]}}}
}\]
where
\begin{align*}
  \choicesVecEq &\triangleq 
  \Case{
\Emit{\Vector{\Var{A}}}{\Constraint{\Var{n}}{\Zero}}{\VNil}
\Emit{\Vector{\Var{A}}}{\Constraint{\Var{n}}{\Suc[\Var{m}]}}
     {\VCons[\PiTel{\Var{m}}{\Nat}
             \PiTel{\Var{a}}{\Var{A}}
             \PiTel{\Var{vs}}{\Vector{\Var{A}}[\Var{m}]}]}
}
\\
  \codeVecEq &\triangleq
\begin{array}[t]{l@{\:}l}
\ReturnDesc & \Collection{
                        \begin{array}{l}
                          \Tag{\CN{nil}},\\ 
                          \Tag{\CN{cons}}
                        \end{array}}
\\
&                    \CollectionElim{
                      \begin{array}{l@{\DoReturn}l}
                        \Tag{\CN{nil}} & \DSigma[(\Var{n} \PropEqual \Zero)\: \Lam{\_}{\DUnit}], \\
                        \Tag{\CN{cons}} & \DSigma[{\Nat\: \Lam{\Var{m}}
                                          \DSigma[{\Var{A}\: \Lam{\_}
                                          \DVar[\Pair{\Void}{\Var{m}}] \DTimes
                                          \DSigma[(\Var{n} \PropEqual \Suc[\Var{m}])\: \Lam{\_}\DUnit]}]}]
                      \end{array}
                    }
\end{array}
\end{align*}

\end{example}

\newcommand{\choicesVecComp}{\mathrm{choices}_{\to}}
\newcommand{\codeVecComp}{\mathrm{code}_{\to}}

\begin{example}[Vector, computed]

The same skeleton is used in the alternative definition of vectors,
but the choices of constructors -- and therefore the resulting code --
are different:
{\small
\begin{align*}
  \choicesVecComp &\triangleq 
  \Case{
\By{\Vector{\Var{A}} & \Var{n}}{\NatCase[\Var{n}]}{
\Emit{\quad\Vector{\Var{A}}}{\Zero}{\VNil}
\Emit{\quad\Vector{\Var{A}}}
     {(\Suc[\Var{m}])}
     {\VCons[\PiTel{\Var{a}}{\Var{A}}
             \PiTel{\Var{vs}}{\Vector{\Var{A}}[\Var{m}]}]}}
}
\\
  \codeVecComp &\triangleq
\begin{array}[t]{l@{\:}l}
\ReturnDesc& \Collection{\Tag{\CN{elim}}}\\
           & \CollectionElim{
              \Code{
              \Tag{\CN{elim}} \DoReturn
              \begin{array}[c]{l}
              \NatCase[{\Var{n}\:
                        (\Lam{\Var{n}}{\LabelDesc{\FatVector[\ParamBracket{\Var{A}}\IndexBracket{\Var{n}}]}})  \\
                      \qquad (\ReturnDesc[\Collection{\Tag{\CN{nil}}}\:
                                 \CollectionElim{\Tag{\CN{nil}} 
                                                   \DoReturn
                                                   \DUnit}])\\
                      \qquad (\Lam{\Var{m}}
                             {\ReturnDesc[\Collection{\Tag{\CN{cons}}}\:
                                      \CollectionElim{\Tag{\CN{cons}}
                                                      \DoReturn
                                                      \DSigma[{\Var{A}\: \Lam{\_}
                                                      \DVar[\Pair{\Void}{\Var{m}}] \DTimes \DUnit}]}]})}]
              \end{array}}}
\end{array}
\end{align*}
}

\end{example}

%% *** <- Elaboration of choice/pattern
%% **** <- Judgment form
%% **** <- Intuition
%% **** <- Soundness
%% **** <- Definition
%% **** <- Remark: from patterns to indices (fig:idesc-elab-indices)
%% ***** <- Refine the data definition problem
%% ****** <- Specify constraints
%% ****** -> Get elaborated at each constructor site

\paragraph{Elaboration of pattern choices (Fig.~\ref{fig:idesc-elab-datapatts}):}
This elaboration process is an extra step that was not necessary with
inductive types. With inductive families, we can both constrain the
index to some particular value or compute over the index to refine the
choice of constructors. Hence, an inductive definition is a list of
pattern choices, potentially ending with a computation over the
indices. Since the case where no index computation is performed is
merely a special case of the rule we give, we save space and do not
write it down explicitly.

The elaboration of pattern choices consists in interpreting the
datatype patterns of each constructor choices. Then, the resulting
labels are used to elaborate these constructors choices. If there is a
computation over indices, we rely on elimination with a motive
\citep{mcbride:elim-2,mcbride:elim} to generate a type theoretic term
from the elimination principle provided by the user. We then interpret
each resulting sub-branch as a pattern choice itself. All in all, this
elaboration step satisfies the following invariant:
\begin{lemma}\label{lemma:idesc-elab-datapatts}

If 
\(\left\{\Code{
\TypeIDataTel{\Gamma}{l} \\
\ElabDataPatts{\Gamma}
                     {\mathrm{patts}}
                     {\LabelDesc{l}}
                     {\mathrm{code}}
}\right.\), then
\(
\TypeJudgment{\Gamma}{\mathrm{code}}{\LabelDesc{l}}
\)

\end{lemma}

Note that we rely on the translation from datatype patterns to data
telescope (Fig.~\ref{fig:idesc-elab-indices}): since we will be
elaborating each individual constructor against these labels, we will
generate the valid equality constraints at the end of each telescope
of arguments.

%% **** Vector, constrained

\begin{example}[Vector, constrained]

The elaboration of datatype patterns simply proceeds over the choices
of constructors, triggering the elaboration of choices on \(\VNil\)
and \(\VCons\):
\[
\Rule{\Code[c]{
      \ElabIndices{\FatVector[\ParamBracket{\Var{A}}\IndexBracket{\Var{n}}]}
                  {\Vector{\Var{A}}[\Constraint{\Var{n}}{\Zero}]}
                  {\FatVector[\ParamBracket{\Var{A}}\Constraint{\Var{n}}{\Zero}]} \\
      \ElabChoices{\TypeAnn{A}{\Set}, \TypeAnn{n}{\Nat}}
                  {\VNil}
                  {\LabelDesc{\FatVector[\ParamBracket{\Var{A}}\Constraint{\Var{n}}{\Zero}]}}
                  {\Collection{\Tag{\CN{nil}}}}
                  {\CollectionElim{\Tag{\CN{nil}} 
                      \DoReturn
                      \DSigma[(\Var{n} \PropEqual \Zero)\: \Lam{\_}{\DUnit}]}}
      \\
      \\
      \ElabIndices{\FatVector[\ParamBracket{\Var{A}}\IndexBracket{\Var{n}}]}
                  {\Vector{\Var{A}}[\Constraint{\Var{n}}{\Suc[\Var{m}]}]}
                  {\FatVector[\ParamBracket{\Var{A}}\Constraint{\Var{n}}{\Suc[\Var{m}]}]} \\
      \Code{
      \ElabChoices{\TypeAnn{A}{\Set}, \TypeAnn{n}{\Nat}}
                  {\VCons[\PiTel{\Var{m}}{\Nat}
                          \PiTel{\Var{a}}{\Var{A}}
                          \PiTel{\Var{vs}}{\Vector{\Var{A}}[\Var{m}]}] \\ \qquad\qquad}
                  {\LabelDesc{\FatVector[\ParamBracket{\Var{A}}\Constraint{\Var{n}}{\Suc[\Var{m}]}]}}
                  {\Collection{\Tag{\CN{cons}}}}
                  {\CollectionElim{\Tag{\CN{cons}} 
                      \DoReturn
                      \DSigma[{\Nat\: \Lam{\Var{m}}
                               \DSigma[{\Var{A}\: \Lam{\_}
                               \DVar[\Pair{\Void}{\Var{m}}] \DTimes
                               \DSigma[(\Var{n} \PropEqual \Suc[\Var{m}])\: \Lam{\_}{\DUnit}]}]}]}}}}}
     {\ElabDataPatts{\TypeAnn{\Var{A}}{\Set},
                     \TypeAnn{\Var{n}}{\Nat}}
                    {[\choicesVecEq]}
                    {\FatVector[\ParamBracket{\Var{A}}\IndexBracket{\Var{n}}]}
                    {[\codeVecEq]}}
\]
where \(\choicesVecEq\) and \(\codeVecEq\) are the same as above.

\end{example}

%% **** Vector, computed

\begin{example}[Vector, computed]

Similarly for the alternative definition of vectors, this triggers the
elaboration of a motive, of the \(\VNil\) and \(\VCons\) patterns:
\[
\Rule{\Code[c]{
    \Code{
      \ElabIndices{\FatVector[\ParamBracket{\Var{A}}\IndexBracket{\Var{n}}]}
                  {\Vector{\Var{A}}[\Var{n}]}
                  {\FatVector[\ParamBracket{\Var{A}}\IndexBracket{\Var{n}}]} \\
      \ElabEWM{\TypeAnn{A}{\Set}, \TypeAnn{n}{\Nat}}
              {\NatCase[\Var{n}] \\\qquad\qquad}
              {\Lam{\Var{ih_0}\: \Var{ih_n}}
                \NatCase[\Var{n}\:
                        (\Lam{\Var{n}}{\LabelDesc{\FatVector[\ParamBracket{\Var{A}}\IndexBracket{\Var{n}}]}})\: 
                        \Var{ih0}\: \Var{ih_n}] \\\qquad\qquad\qquad\qquad}
              {\LabelDesc{\FatVector[\ParamBracket{\Var{A}}\IndexBracket{\Zero}]} \To
               (\PiTo{\Var{m}}{\Nat} \LabelDesc{\FatVector[\ParamBracket{\Var{A}}\IndexBracket{\Suc[\Var{m}]}]}) \To
               \LabelDesc{\FatVector[\ParamBracket{\Var{A}}\IndexBracket{\Var{n}}]}}} 
      \\
      \ElabDataPatts{\TypeAnn{\Var{A}}{\Set},
                     \TypeAnn{\Var{n}}{\Nat}}
                    {\VNil}
                    {\FatVector[\ParamBracket{\Var{A}}\IndexBracket{\Zero}]}
                    {\ReturnDesc[\Collection{\Tag{\CN{nil}}}\:
                                 \CollectionElim{\Tag{\CN{nil}}
                                                 \DoReturn
                                                 \DUnit}]}
      \\
      \Code{
      \ElabDataPatts{\TypeAnn{\Var{A}}{\Set},
                     \TypeAnn{\Var{n}}{\Nat},
                     \TypeAnn{\Var{m}}{\Nat}}
                    {\VCons[\PiTel{\Var{a}}{\Var{A}}
                            \PiTel{\Var{vs}}{\Vector{\Var{A}}[\Var{m}]}] \\\qquad\qquad}
                    {\FatVector[\ParamBracket{\Var{A}}\IndexBracket{\Suc[\Var{m}]}]}
                    {\ReturnDesc[\Collection{\Tag{\CN{cons}}}\:
                                 \CollectionElim{\Tag{\CN{cons}}
                                                 \DoReturn
                                                 \DSigma[{\Var{A}\: \Lam{\_}
                                                 \DVar[\Pair{\Void}{\Var{m}}] \DTimes \DUnit}]}]}}
      }}
     {\ElabDataPatts{\TypeAnn{\Var{A}}{\Set},
                     \TypeAnn{\Var{n}}{\Nat}}
                    {[\choicesVecComp]}
                    {\FatVector[\ParamBracket{\Var{A}}\IndexBracket{\Var{n}}]}
                    {[\codeVecComp]}}
\]
with \(\choicesVecComp\) and \(\codeVecComp\) defined above.
In turn, this triggers the elaboration of datatype choices for the
\(\VNil\) and \(\VCons\) patterns:
\[
\Rule{\ElabChoices{\TypeAnn{A}{\Set}, \TypeAnn{n}{\Nat}}
                  {\VNil}
                  {\LabelDesc{\FatVector[\ParamBracket{\Var{A}}\IndexBracket{\Zero}]}}
                  {\Collection{\Tag{\CN{nil}}}}
                  {\CollectionElim{\Tag{\CN{nil}} \DoReturn \DUnit}}}
     {\ElabDataPatts{\TypeAnn{\Var{A}}{\Set},
                     \TypeAnn{\Var{n}}{\Nat}}
                    {\VNil}
                    {\FatVector[\ParamBracket{\Var{A}}\IndexBracket{\Zero}]}
                    {\ReturnDesc[\Collection{\Tag{\CN{nil}}}\:
                                 \CollectionElim{\Tag{\CN{nil}}
                                                 \DoReturn
                                                 \DUnit}]}}
\]

\[
\Rule{\Code{
      \ElabChoices{\TypeAnn{A}{\Set}, \TypeAnn{n}{\Nat}, \TypeAnn{m}{\Nat}}
                  {\VCons[{\PiTel{\Var{a}}{\Var{A}}
                           \PiTel{\Var{vs}}{\Vector{\Var{A}}[\Var{m}]}}] \\\qquad\qquad}
                  {\LabelDesc{\FatVector[\ParamBracket{\Var{A}}\IndexBracket{\Suc[\Var{m}]}]}}
                  {\Collection{\Tag{\CN{cons}}}}
                  {\CollectionElim{\Tag{\CN{cons}} 
                      \DoReturn
                      \DSigma[{\Var{A}\: \Lam{\_}
                      \DVar[\Pair{\Void}{\Var{m}}] \DTimes \DUnit}]}}}}
     {\Code{
      \ElabDataPatts{\TypeAnn{\Var{A}}{\Set},
                     \TypeAnn{\Var{n}}{\Nat},
                     \TypeAnn{\Var{m}}{\Nat}}
                    {\VCons[\PiTel{\Var{a}}{\Var{A}}
                            \PiTel{\Var{vs}}{\Vector{\Var{A}}[\Var{m}]}] \\\qquad\qquad}
                    {\FatVector[\ParamBracket{\Var{A}}\IndexBracket{\Suc[\Var{m}]}]}
                    {\ReturnDesc[\Collection{\Tag{\CN{cons}}}\:
                                 \CollectionElim{\Tag{\CN{cons}}
                                                 \DoReturn
                                                 \DSigma[{\Var{A}\: \Lam{\_}
                                                 \DVar[\Pair{\Void}{\Var{m}}] \DTimes \DUnit}]}]}}}
\]

\end{example}

%% *** <- Elaboration of choices
%% **** <- Judgment form
%% **** <- Intuition
%% **** <- Soundness
%% **** <- Definition

\paragraph{Elaboration of choices (Fig.~\ref{fig:idesc-elab-choices}):}
The elaboration of the choice of datatypes is the same as for
inductive types: we merely collect the tag and code of each individual
constructor and return these as enumerations. This step is subject to
the following soundness property:
\begin{lemma}\label{lemma:idesc-elab-choices}
If 
\(\left\{\Code{
\TypeIDataTel{\Gamma}{l} \\
\ElabChoices{\Gamma}
                   {\mathrm{choices}}
                   {\LabelDesc{l}}
                   {E}
                   {T}
}\right.\), then
\(\left\{\Code{
\TypeJudgment{\Gamma}{E}{\EnumU} \\
\TypeJudgment{\Gamma}{T}{\EnumT[\Meta{ts}] \To \IDesc[\InterpretIDataTel{\Meta{l}}]}
}\right.\)

\end{lemma}

%% **** Vector, constrained

\begin{example}[Vector, constrained]

In the particular example of vector, there is only one choice of
constructor when \(\Constraint{\Var{n}}{\Zero}\) -- namely \(\VNil\)
-- and when \(\Constraint{\Var{n}}{\Suc[\Var{m}]}\) -- namely
\(\VCons\). Therefore, we obtain the elaboration of choices from the
elaboration of the unique constructor, in both cases:
\[\Code{
\Rule{\ElabConstr{\TypeAnn{A}{\Set}, \TypeAnn{n}{\Nat}}
                 {\VNil}
                 {\LabelDesc{\FatVector[\ParamBracket{\Var{A}}\Constraint{\Var{n}}{\Zero}]}}
                 {\Tag{\CN{nil}}}
                 {\DSigma[(\Var{n} \PropEqual \Zero)\: \Lam{\_}{\DUnit}]}}
     {\ElabChoices{\TypeAnn{A}{\Set}, \TypeAnn{n}{\Nat}}
                  {\VNil}
                  {\LabelDesc{\FatVector[\ParamBracket{\Var{A}}\Constraint{\Var{n}}{\Zero}]}}
                  {\Collection{\Tag{\CN{nil}}}}
                  {\CollectionElim{\Tag{\CN{nil}} 
                      \DoReturn
                      \DSigma[(\Var{n} \PropEqual \Zero)\: \Lam{\_}{\DUnit}]}}}
\\
\\
\Rule{\Code{
      \ElabConstr{\TypeAnn{A}{\Set}, \TypeAnn{n}{\Nat}}
                 {\VCons[\PiTel{\Var{m}}{\Nat}
                         \PiTel{\Var{a}}{\Var{A}}
                         \PiTel{\Var{vs}}{\Vector{\Var{A}}[\Var{m}]}] \\ \qquad\qquad}
                 {\LabelDesc{\FatVector[\ParamBracket{\Var{A}}\Constraint{\Var{n}}{\Suc[\Var{m}]}]}}
                 {\Tag{\CN{cons}}}
                 {\DSigma[{\Nat\: \Lam{\Var{m}}
                           \DSigma[{\Var{A}\: \Lam{\_}
                           \DVar[\Pair{\Void}{\Var{m}}] \DTimes
                           \DSigma[(\Var{n} \PropEqual \Suc[\Var{m}])\: \Lam{\_}\DUnit]}]}]}}}
     {\Code{
      \ElabChoices{\TypeAnn{A}{\Set}, \TypeAnn{n}{\Nat}}
                  {\VCons[\PiTel{\Var{m}}{\Nat}
                          \PiTel{\Var{a}}{\Var{A}}
                          \PiTel{\Var{vs}}{\Vector{\Var{A}}[\Var{m}]}] \\ \qquad\qquad}
                  {\LabelDesc{\FatVector[\ParamBracket{\Var{A}}\Constraint{\Var{n}}{\Suc[\Var{m}]}]}}
                  {\Collection{\Tag{\CN{cons}}}}
                  {\CollectionElim{\Tag{\CN{cons}} 
                      \DoReturn
                      \DSigma[{\Nat\: \Lam{\Var{m}}
                               \DSigma[{\Var{A}\: \Lam{\_}
                               \DVar[\Pair{\Void}{\Var{m}}] \DTimes
                               \DSigma[(\Var{n} \PropEqual \Suc[\Var{m}])\: \Lam{\_}\DUnit]}]}]}}}}
}\]

\end{example}

%% **** Vector, computed

\begin{example}[Vector, computed]

The same situation arises in the alternative definition of vector:
once we have determined which index we are dealing with, there is a
single constructor available. Hence, we move from the elaboration of
choices to the elaboration of the unique constructor:
\[\Code{
\Rule{\ElabConstr{\TypeAnn{A}{\Set}, \TypeAnn{n}{\Nat}}
                 {\VNil}
                 {\LabelDesc{\FatVector[\ParamBracket{\Var{A}}\IndexBracket{\Zero}]}}
                 {\Tag{\CN{nil}}}
                 {\DUnit}}
     {\ElabChoices{\TypeAnn{A}{\Set}, \TypeAnn{n}{\Nat}}
                  {\VNil}
                  {\LabelDesc{\FatVector[\ParamBracket{\Var{A}}\IndexBracket{\Zero}]}}
                  {\Collection{\Tag{\CN{nil}}}}
                  {\CollectionElim{\Tag{\CN{nil}} \DoReturn \DUnit}}}
\\
\\
\Rule{\Code{
      \ElabConstr{\TypeAnn{A}{\Set}, \TypeAnn{n}{\Nat}, \TypeAnn{m}{\Nat}}
                 {\VCons[{\PiTel{\Var{a}}{\Var{A}}
                          \PiTel{\Var{vs}}{\Vector{\Var{A}}[\Var{m}]}}] \\\qquad\qquad}
                 {\LabelDesc{\FatVector[\ParamBracket{\Var{A}}\IndexBracket{\Suc[\Var{m}]}]}}
                 {\Tag{\CN{cons}}}
                 {\DSigma[{\Var{A}\: \Lam{\_}
                  \DVar[\Pair{\Void}{\Var{m}}] \DTimes \DUnit}]}}}
     {\Code{
      \ElabChoices{\TypeAnn{A}{\Set}, \TypeAnn{n}{\Nat}, \TypeAnn{m}{\Nat}}
                  {\VCons[{\PiTel{\Var{a}}{\Var{A}}
                           \PiTel{\Var{vs}}{\Vector{\Var{A}}[\Var{m}]}}] \\\qquad\qquad}
                  {\LabelDesc{\FatVector[\ParamBracket{\Var{A}}\IndexBracket{\Suc[\Var{m}]}]}}
                  {\Collection{\Tag{\CN{cons}}}}
                  {\CollectionElim{\Tag{\CN{cons}} 
                      \DoReturn
                      \DSigma[{\Var{A}\: \Lam{\_}
                      \DVar[\Pair{\Void}{\Var{m}}] \DTimes \DUnit}]}}}}
}\]

\end{example}

%% *** <- Elaboration of constructor
%% **** <- Judgment form
%% **** <- Intuition
%% **** <- Soundness
%% **** <- Definition

\paragraph{Elaboration of constructor (Fig.~\ref{fig:idesc-elab-constr}):}
Again, the elaboration of constructor is not any different from the
inductive types case. It is subject to the following invariant:
\begin{lemma}\label{lemma:idesc-elab-constr}
If
\(\left\{\Code{
\TypeIDataTel{\Gamma}{l} \\
\ElabConstr{\Gamma}
           {c}
           {\LabelDesc{l}}
           {t}
           {\mathrm{code}}
}\right.\), then
\(\left\{\Code{
\TypeJudgment{\Gamma}{t}{\UId} \\
\TypeJudgment{\Gamma}{\mathrm{code}}{\IDesc[\InterpretIDataTel{l}]}
}\right.\)
\end{lemma}

%% **** Vector, constrained

\begin{example}[Vector, constrained]

Elaboration of the constructors is straightforward, simply switching
to the elaboration of the arguments:
\[\Code[c]{
\Rule{\ElabArg{\TypeAnn{A}{\Set}, \TypeAnn{n}{\Nat}}
        {\epsilon}
        {\LabelDesc{\FatVector[\ParamBracket{\Var{A}}\: \Constraint{\Var{n}}{\Zero}]}}
        {\DSigma[(\Var{n} \PropEqual \Zero)\: \Lam{\_}{\DUnit}]}}
     {\ElabConstr{\TypeAnn{A}{\Set}, \TypeAnn{n}{\Nat}}
                 {\VNil}
                 {\LabelDesc{\FatVector[\ParamBracket{\Var{A}}\: \Constraint{\Var{n}}{\Zero}]}}
                 {\Tag{\CN{nil}}}
                 {\DSigma[(\Var{n} \PropEqual \Zero)\: \Lam{\_}{\DUnit}]}}
\\
\\
\Rule{\Code{
      \ElabArg{\TypeAnn{A}{\Set}, \TypeAnn{n}{\Nat}}
              {\PiTel{\Var{m}}{\Nat}
               \PiTel{\Var{a}}{\Var{A}}
               \PiTel{\Var{vs}}{\Vector{\Var{A}}[\Var{m}]} \\ \qquad\qquad}
              {\LabelDesc{\FatVector[\ParamBracket{\Var{A}}\: \Constraint{\Var{n}}{\Suc[\Var{m}]}]}}
              {\DSigma[{\Nat\: \Lam{\Var{m}}
               \DSigma[{\Var{A}\: \Lam{\_}
               \DVar[\Pair{\Void}{\Var{m}}] \DTimes
               \DSigma[(\Var{n} \PropEqual \Suc[\Var{m}])\: \Lam{\_}\DUnit]}]}]}}}
     {\Code{
      \ElabConstr{\TypeAnn{A}{\Set}, \TypeAnn{n}{\Nat}}
                 {\VCons[\PiTel{\Var{m}}{\Nat}
                        \PiTel{\Var{a}}{\Var{A}}
                        \PiTel{\Var{vs}}{\Vector{\Var{A}}[\Var{m}]}]  \\ \qquad\qquad}
                 {\LabelDesc{\FatVector[\ParamBracket{\Var{A}}\: \Constraint{\Var{n}}{\Suc[\Var{m}]}]}}
                 {\Tag{\CN{cons}}}
                 {\DSigma[{\Nat\: \Lam{\Var{m}}
                  \DSigma[{\Var{A}\: \Lam{\_}
                  \DVar[\Pair{\Void}{\Var{m}}] \DTimes
                  \DSigma[(\Var{n} \PropEqual \Suc[\Var{m}])\: \Lam{\_}\DUnit]}]}]}}}
}\]

\end{example}

%% **** Vector, computed

\begin{example}[Vector, computed]

Similarly for the alternative definition, we have:
\[\Code[c]{
\Rule{\ElabArg{\TypeAnn{A}{\Set}, \TypeAnn{n}{\Nat}}
              {\epsilon}
              {\LabelDesc{\FatVector[\ParamBracket{\Var{A}}\: \IndexBracket{\Zero}]}}
              {\DUnit}}
     {\ElabConstr{\TypeAnn{A}{\Set}, \TypeAnn{n}{\Nat}}
                 {\VNil}
                 {\LabelDesc{\FatVector[\ParamBracket{\Var{A}}\IndexBracket{\Zero}]}}
                 {\Tag{\CN{nil}}}
                 {\DUnit}}
\\
\\
\Rule{\ElabArg{\TypeAnn{A}{\Set}, \TypeAnn{n}{\Nat}, \TypeAnn{m}{\Nat}}
              {\PiTel{\Var{a}}{\Var{A}}
               \PiTel{\Var{vs}}{\Vector{\Var{A}}[\Var{m}]}}
              {\LabelDesc{\FatVector[\ParamBracket{\Var{A}}\: \IndexBracket{\Suc[\Var{m}]}]}}
              {\DSigma[{\Var{A}\: \Lam{\_}
               \DVar[\Pair{\Void}{\Var{m}}] \DTimes \DUnit}]}}
     {\Code{
      \ElabConstr{\TypeAnn{A}{\Set}, \TypeAnn{n}{\Nat}, \TypeAnn{m}{\Nat}}
                 {\VCons[\PiTel{\Var{a}}{\Var{A}}
                         \PiTel{\Var{vs}}{\Vector{\Var{A}}[\Var{m}]}]\\ \qquad\qquad}
                 {\LabelDesc{\FatVector[\ParamBracket{\Var{A}}\IndexBracket{\Suc[\Var{m}]}]}}
                 {\Tag{\CN{cons}}}
                 {\DSigma[{\Var{A}\: \Lam{\_}
                  \DVar[\Pair{\Void}{\Var{m}}] \DTimes \DUnit}]}}}
}\]

\end{example}

%% *** <- Elaboration of arguments
%% **** <- Judgment form
%% **** <- Intuition
%% **** <- Soundness
%% **** <- Definition

\paragraph{Elaboration of arguments (Fig.~\ref{fig:idesc-elab-arg}):}
The elaboration of arguments follows the same principle as for
inductive types. However, when elaborating a recursive argument, we
must extract the indices for which that recursive step is
taken. Besides, after having encoded the list of arguments, we must
switch to translating the potential equality constraints, which are
dictacted by the label type we are elaborating against. This step is
subject to the following soundness property:
\begin{lemma}\label{lemma:idesc-elab-arg}
If 
\(\left\{\Code{
  \TypeIDataTel{\Gamma}{\Meta{l}} \\
  \ElabArg{\Gamma}
          {\mathrm{args}}
          {\LabelDesc{l}}
          {\mathrm{code}}
}\right.\), then
\(
\TypeJudgment{\Gamma}{\mathrm{code}}{\IDesc[\InterpretIDataTel{\Meta{l}}]}
\)
\end{lemma}

%% **** Vector, constrained

\begin{example}[Vector, constrained]

We then elaborate the arguments by unfolding the telescope, at which
point we switch to elaborating the constraints. We do so immediatly in
the \(\VNil\) case:
\[
\Rule{\ElabEqs{\TypeAnn{A}{\Set}, \TypeAnn{n}{\Nat}}
              {\FatVector[\ParamBracket{\Var{A}}\: \Constraint{\Var{n}}{\Zero}]}
              {\DSigma[(\Var{n} \PropEqual \Zero)\: \Lam{\_}{\DUnit}]}}
     {\ElabArg{\TypeAnn{A}{\Set}, \TypeAnn{n}{\Nat}}
              {\epsilon}
              {\LabelDesc{\FatVector[\ParamBracket{\Var{A}}\: \Constraint{\Var{n}}{\Zero}]}}
              {\DSigma[(\Var{n} \PropEqual \Zero)\: \Lam{\_}{\DUnit}]}}
\]
While a few steps are necessary to elaborate the arguments in the
\(\VCons\) case, including the elaboration of the recursive call:
\[
\Rule{
  \Rule{
    \Rule{
      \Rule{
      \ElabEqs{\TypeAnn{A}{\Set}, \TypeAnn{n}{\Nat}, \TypeAnn{m}{\Nat}}
              {\FatVector[\ParamBracket{\Var{A}}\: \Constraint{\Var{n}}{\Suc[\Var{m}]}]}
              {\DSigma[(\Var{n} \PropEqual \Suc[\Var{m}])\: \Lam{\_}{\DUnit}]}}
           {\Code{
      \ElabRecArgs{\TypeAnn{A}{\Set}, \TypeAnn{n}{\Nat}, \TypeAnn{m}{\Nat}}
                  {\Vector{\Var{A}}[\Var{m}]}
                  {\FatVector[\ParamBracket{\Var{A}}\: \Constraint{\Var{n}}{\Suc[\Var{m}]}]}
                  {\Pair{\Void}{\Var{m}}}
      \\
      \ElabArg{\TypeAnn{A}{\Set}, \TypeAnn{n}{\Nat}, \TypeAnn{m}{\Nat}}
              {\epsilon}
              {\LabelDesc{\FatVector[\ParamBracket{\Var{A}}\: \Constraint{\Var{n}}{\Suc[\Var{m}]}]}}
              {\DSigma[(\Var{n} \PropEqual \Suc[\Var{m}])\: \Lam{\_}\DUnit]}}}}
         {
      \ElabArg{\TypeAnn{A}{\Set}, \TypeAnn{n}{\Nat}, \TypeAnn{m}{\Nat}}
              {\PiTel{\Var{vs}}{\Vector{\Var{A}}[\Var{m}]}}
              {\LabelDesc{\FatVector[\ParamBracket{\Var{A}}\: \Constraint{\Var{n}}{\Suc[\Var{m}]}]}}
              {\DVar[\Pair{\Void}{\Var{m}}] \DTimes
               \DSigma[(\Var{n} \PropEqual \Suc[\Var{m}])\: \Lam{\_}\DUnit]}}}
       {\Code{
      \ElabArg{\TypeAnn{A}{\Set}, \TypeAnn{n}{\Nat}, \TypeAnn{m}{\Nat}}
              {\PiTel{\Var{a}}{\Var{A}}
               \PiTel{\Var{vs}}{\Vector{\Var{A}}[\Var{m}]} \\ \qquad\qquad}
              {\LabelDesc{\FatVector[\ParamBracket{\Var{A}}\: \Constraint{\Var{n}}{\Suc[\Var{m}]}]}}
              {\DSigma[{\Var{A}\: \Lam{\_}
               \DVar[\Pair{\Void}{\Var{m}}] \DTimes
               \DSigma[(\Var{n} \PropEqual \Suc[\Var{m}])\: \Lam{\_}\DUnit]}]}}}}
     {\Code{
      \ElabArg{\TypeAnn{A}{\Set}, \TypeAnn{n}{\Nat}}
              {\PiTel{\Var{m}}{\Nat}
               \PiTel{\Var{a}}{\Var{A}}
               \PiTel{\Var{vs}}{\Vector{\Var{A}}[\Var{m}]} \\ \qquad\qquad}
              {\LabelDesc{\FatVector[\ParamBracket{\Var{A}}\: \Constraint{\Var{n}}{\Suc[\Var{m}]}]}}
              {\DSigma[{\Nat\: \Lam{\Var{m}}
               \DSigma[{\Var{A}\: \Lam{\_}
               \DVar[\Pair{\Void}{\Var{m}}] \DTimes
               \DSigma[(\Var{n} \PropEqual \Suc[\Var{m}])\: \Lam{\_}\DUnit]}]}]}}}
\]

\end{example}

%% **** Vector, computed

\begin{example}[Vector, computed]

In the alternative definition, the process is similar:
\[\Code[c]{
\Rule{\ElabEqs{\TypeAnn{A}{\Set}, \TypeAnn{n}{\Nat}}
              {\FatVector[\ParamBracket{\Var{A}}\IndexBracket{\Zero}]}
              {\DUnit}}
     {\ElabArg{\TypeAnn{A}{\Set}, \TypeAnn{n}{\Nat}}
              {\epsilon}
              {\LabelDesc{\FatVector[\ParamBracket{\Var{A}}\IndexBracket{\Zero}]}}
              {\DUnit}}
\\
\\
\Rule{
  \Rule{
    \Rule{
      \ElabEqs{\TypeAnn{A}{\Set}, \TypeAnn{n}{\Nat}, \TypeAnn{m}{\Nat}}
              {\FatVector[\ParamBracket{\Var{A}}\IndexBracket{\Suc[\Var{m}]}]}
              {\DUnit}}
     {\Code{
      \ElabRecArgs{\TypeAnn{A}{\Set}, \TypeAnn{n}{\Nat}, \TypeAnn{m}{\Nat}}
            {\Vector{\Var{A}}[\Var{m}]}
            {\FatVector[\ParamBracket{\Var{A}}\IndexBracket{\Suc[\Var{m}]}]}
            {\Pair{\Void}{\Var{m}}}

      \\
      \ElabArg{\TypeAnn{A}{\Set}, \TypeAnn{n}{\Nat}, \TypeAnn{m}{\Nat}}
              {\epsilon}
              {\LabelDesc{\FatVector[\ParamBracket{\Var{A}}\IndexBracket{\Suc[\Var{m}]}]}}
              {\DUnit}}}}
     {\ElabArg{\TypeAnn{A}{\Set}, \TypeAnn{n}{\Nat}, \TypeAnn{m}{\Nat}}
              {\PiTel{\Var{vs}}{\Vector{\Var{A}}[\Var{m}]}}
              {\LabelDesc{\FatVector[\ParamBracket{\Var{A}}\IndexBracket{\Suc[\Var{m}]}]}}
              {\DVar[\Pair{\Void}{\Var{m}}] \DTimes \DUnit}}}
     {\ElabArg{\TypeAnn{A}{\Set}, \TypeAnn{n}{\Nat}, \TypeAnn{m}{\Nat}}
              {\PiTel{\Var{a}}{\Var{A}}
               \PiTel{\Var{vs}}{\Vector{\Var{A}}[\Var{m}]}}
              {\LabelDesc{\FatVector[\ParamBracket{\Var{A}}\IndexBracket{\Suc[\Var{m}]}]}}
              {\DSigma[{\Var{A}\: \Lam{\_}
               \DVar[\Pair{\Void}{\Var{m}}] \DTimes \DUnit}]}}
}\]

\end{example}

%% *** <- Elaboration of equality constraints
%% **** <- Judgment form
%% **** <- Intuition
%% **** <- Soundness
%% **** <- Definition

\paragraph{Elaboration of constraints (Fig.~\ref{fig:idesc-elab-eqs}):}
In order to generate the equality constraints, we simply traverse the
label we are elaborating against. On constraints, we generate the
corresponding equality constraint, using whatever propositional
equality the system has to offer. On parameters and (unconstrained)
indices, we simply go through. This step satisfies the following
property:
\begin{lemma}\label{lemma:idesc-elab-eqs}
If 
\(\left\{\Code{
\TypeIDataTel{\Gamma}{l} \\
\ElabEqs{\Gamma}
        {l}
        {q}
}\right.\), then
\(\TypeJudgment{\Gamma}{q}{\IDesc[\InterpretIDataTel{l}]}\).

\end{lemma}

%% **** Vector, constrained

\begin{example}[Vector, constrained]

We elaborate the constraints in the \(\VNil\) case -- constraining
\(n\) to \(\Zero\) -- and in the \(\VCons\) case -- constraining \(n\)
to \(\Suc[m]\):
\[
\Rule{
  \Rule{
      \ElabEqs{\TypeAnn{A}{\Set}, \TypeAnn{n}{\Nat}}
              {\FatVector}
              {\DUnit}}{
      \ElabEqs{\TypeAnn{A}{\Set}, \TypeAnn{n}{\Nat}}
              {\FatVector[\ParamBracket{\Var{A}}]}
              {\DUnit}}}
     {\ElabEqs{\TypeAnn{A}{\Set}, \TypeAnn{n}{\Nat}}
              {\FatVector[\ParamBracket{\Var{A}}\: \Constraint{\Var{n}}{\Zero}]}
              {\DSigma[(\Var{n} \PropEqual \Zero)\: \Lam{\_}{\DUnit}]}}
\qquad
\Rule{
  \Rule{
      \ElabEqs{\TypeAnn{A}{\Set}, \TypeAnn{n}{\Nat}, \TypeAnn{m}{\Nat}}
              {\FatVector}
              {\DUnit}}
       {
      \ElabEqs{\TypeAnn{A}{\Set}, \TypeAnn{n}{\Nat}, \TypeAnn{m}{\Nat}}
              {\FatVector[\ParamBracket{\Var{A}}]}
              {\DUnit}}}
     {\ElabEqs{\TypeAnn{A}{\Set}, \TypeAnn{n}{\Nat}, \TypeAnn{m}{\Nat}}
              {\FatVector[\ParamBracket{\Var{A}}\: \Constraint{\Var{n}}{\Suc[\Var{m}]}]}
              {\DSigma[(\Var{n} \PropEqual \Suc[\Var{m}])\: \Lam{\_}{\DUnit}]}}
\]

\end{example}

%% **** Vector, computed

\begin{example}[Vector, computed]

No equations are generated and, indeed, needed for the alternative
definition of vectors:
\[
\Rule{
  \Rule{\ElabEqs{\TypeAnn{A}{\Set}, \TypeAnn{n}{\Nat}}
                {\FatVector}
                {\DUnit}}
     {\ElabEqs{\TypeAnn{A}{\Set}, \TypeAnn{n}{\Nat}}
              {\FatVector[\ParamBracket{\Var{A}}]}
              {\DUnit}}}
     {\ElabEqs{\TypeAnn{A}{\Set}, \TypeAnn{n}{\Nat}}
              {\FatVector[\ParamBracket{\Var{A}}\IndexBracket{\Zero}]}
              {\DUnit}}
\qquad
\Rule{
  \Rule{
      \ElabEqs{\TypeAnn{A}{\Set}, \TypeAnn{n}{\Nat}, \TypeAnn{m}{\Nat}}
              {\FatVector}
              {\DUnit}}
     {\ElabEqs{\TypeAnn{A}{\Set}, \TypeAnn{n}{\Nat}, \TypeAnn{m}{\Nat}}
              {\FatVector[\ParamBracket{\Var{A}}]}
              {\DUnit}}}
     {\ElabEqs{\TypeAnn{A}{\Set}, \TypeAnn{n}{\Nat}, \TypeAnn{m}{\Nat}}
              {\FatVector[\ParamBracket{\Var{A}}\IndexBracket{\Suc[\Var{m}]}]}
              {\DUnit}}
\]

\end{example}

%% *** <- Extraction of indices
%% **** <- Judgment form
%% **** <- Intuition
%% **** <- Soundness
%% **** <- Definition

\paragraph{Extraction of indices (Fig.~\ref{fig:idesc-elab-recargs}):}
On the recursive calls, we must extract the indices at which the call
is made. To do so, we match the type definition with the datatype
label. Parameters are ignored while indices are paired together. On
the datatype name, we inhabit \(\Unit\). This ensures the following
soundness property:
\begin{lemma}\label{lemma:idesc-elab-recargs}
If 
\(\left\{\Code{
  \TypeIDataTel{\Gamma}{l} \\
  \ElabRecArgs{\Gamma}{T}{l}{is}
}\right.\), then
\(\TypeJudgment{\Gamma}{is}{\InterpretIDataTel{l}}\).
\end{lemma}

%% **** <- Example: vector, constrained

\begin{example}[Vector, constrained]

There is only one instance of recursive definition, in the \(\VCons\)
case. Its elaboration goes as follow:
\[
\Rule{
  \Rule{
      \ElabRecArgs{\TypeAnn{A}{\Set}, \TypeAnn{n}{\Nat}, \TypeAnn{m}{\Nat}}
                  {\Vector{}}
                  {\FatVector}
                  {\Void}}
     {\ElabRecArgs{\TypeAnn{A}{\Set}, \TypeAnn{n}{\Nat}, \TypeAnn{m}{\Nat}}
                  {\Vector{\Var{A}}}
                  {\FatVector[\ParamBracket{\Var{A}}]}
                  {\Void}}}
     {\ElabRecArgs{\TypeAnn{A}{\Set}, \TypeAnn{n}{\Nat}, \TypeAnn{m}{\Nat}}
                  {\Vector{\Var{A}}[\Var{m}]}
                  {\FatVector[\ParamBracket{\Var{A}}\: \Constraint{\Var{n}}{\Suc[\Var{m}]}]}
                  {\Pair{\Void}{\Var{m}}}}
\]

\end{example}

%% **** <- Example: vector, computed

\begin{example}[Vector, computed]

Similarly, the elaboration of the index for the recursive definition
is as follow:
\[
\Rule{
  \Rule{
      \ElabRecArgs{\TypeAnn{A}{\Set}, \TypeAnn{n}{\Nat}, \TypeAnn{m}{\Nat}}
                  {\Vector{}}
                  {\FatVector}
                  {\Void}}
     {\ElabRecArgs{\TypeAnn{A}{\Set}, \TypeAnn{n}{\Nat}, \TypeAnn{m}{\Nat}}
                  {\Vector{\Var{A}}}
                  {\FatVector[\ParamBracket{\Var{A}}]}
                  {\Void}}}
     {\ElabRecArgs{\TypeAnn{A}{\Set}, \TypeAnn{n}{\Nat}, \TypeAnn{m}{\Nat}}
                  {\Vector{\Var{A}}[\Var{m}]}
                  {\FatVector[\ParamBracket{\Var{A}}\IndexBracket{\Suc[\Var{m}]}]}
                  {\Pair{\Void}{\Var{m}}}}
\]

\end{example}

%% ** <- Example [in the Spec.]
%% *** <- Nat
%% *** <- Vector (constrained)
%% *** <- Vector (computation)
%% ** <- Soundness
%% *** <- Lemma: extraction of recursive type is "correct"
%% *** <- Lemma: extraction of indices is "correct"
%% *** <- Lemma: elaboration of equality constraints is sound
%% *** <- Lemma: elaboration of arguments is sound
%% *** <- Lemma: elaboration of constructor is sound
%% *** <- Lemma: elaboration of choices is sound
%% *** <- Lemma: elaboration of pattern is sound
%% *** <- Theorem: elaboration of theorem is sound

Having stated the soundness properties of each individual elaboration
steps, we can now state and prove the soundness of the elaboration of
inductive families:
\begin{theorem}[Soundness of elaboration]

If 
\(\ElabIData{\Gamma}
            {D}
            {\overrightarrow{\Param{\Var{p}}{P}}}
            {\overrightarrow{\Index{\Var{i}}{I}}}
            {\Set}
            {\mathrm{choices}}
            {\Delta}\), then
\(\ContextValid{\Delta}\).

\end{theorem}
\begin{proof}

First, we prove that labels elaborate to a type correct index
(Lemma~\ref{lemma:idesc-elab-recargs}). We then prove that the
constraints generated by interpreting the label are valid descriptions
(Lemma~\ref{lemma:idesc-elab-eqs}). From these lemmas, we can prove
the soundness of the elaboration of arguments by induction over the
telescope of arguments (Lemma~\ref{lemma:idesc-elab-arg}). This gives
straightforwardly the validity of the elaboration of a constructor
(Lemma~\ref{lemma:idesc-elab-constr}). Using that Lemma over each
constructors, we thus obtain the soundness of the elaboration of
choices (Lemma~\ref{lemma:idesc-elab-choices}). Using this result and
assuming the soundness of elimination with a motive, we prove
Lemma~\ref{lemma:idesc-elab-datapatts}. This gives the desired result.

\end{proof}

%% ** <- Discussion 
%% *** <- Reduction to Coq definition
%% **** <- For subset without index-computation
%% **** <- Conjecture: equivalence with Coq induction principles
%% * Discussion

\section{Reflections on Inductives}
\label{sec:discussion}

%% ** <- Motivation: Explore our possibilities
%% *** <- Syntactic-only presentation stuck at meta-level
%% **** -> All support for inductives has to be provided meta
%% **** /> Or, tricky quote/unquote interface
%% *** <- Reflection of inductives!
%% **** -> Metatheory of inductives is a universe
%% **** -> Can implement what was metatheory
%% **** -> Describe two possible extensions
%% ***** -> Reflected constructions on constructor
%% ****** -> For implementer
%% ****** -> Spend more time implementing the type theory in type theory
%% ******* -> Step toward bootstrapping
%% ***** -> Generic deriving 
%% ****** -> For programmers
%% ****** -> Spend less time waiting for features

Having described our infrastructure to elaborate inductive definitions
down to descriptions, we would like to give an overview of the
possibilities offered by such a system. Indeed, in a purely syntactic
presentation of inductives, we are stuck at the meta-level of the type
theory: if we want to provide support to manipulate inductive types,
it must be implemented as part of the theorem prover, out of the type
theory. Alternatively, a quoting/unquoting mechanism could be provided
but guaranteeing the safety of such an extension is likely to be
tricky.

Because our type theory reflects inductives in itself, the meta-theory
of inductive types is no more than a universe. What used to be
meta-theoretical constructions can now be implemented from within the
type theory, benefiting from the various amenities offered by a
dependently-typed programming language. In this Section, we present
two examples of such ``reflection on inductives''. Our first example,
reflecting constructions on
constructors~\citep{mcbride:construction-constructor}, will appeal to
the implementers: we hint at the possibility of implementing key
features of the type theory within itself, a baby step toward
bootstrapping. Our second example, providing a user defined
\texttt{deriving} mechanism, should appeal to programmers: we
illustrate how programmers could provide generic operations over
datatypes and see them automatically integrated in their development.
For simplicity and conciseness, we shall define these mechanisms over
our universe of inductive types, \(\Desc\). Nonetheless, it is
straightforward, but more verbose, to extend these constructions to
inductive families.

%% ** <- Construction on constructors
%% *** <- Cf. McBride-McKinna, Gimenez
%% *** <- Compute specialized lemmas on inductives
%% **** <- Generic operation specialized to definition
%% **** -> Convenience
%% *** -> Relection of inductives
%% **** <- McBride-McKinna on dotdotdot syntax
%% **** -> Work on code

\subsection{A few constructions on constructors, internalized}
\label{sec:const-on-const}

\citet{mcbride:construction-constructor} describe a collection of
lemmas that theorem prover's implementer would like to export with
every inductive type. In that paper, the authors first show how one
can reduce case analysis and course-of-value recursion to standard
induction. Then, they describe two lemmas over datatype constructors:
\emph{no confusion} -- constructors are injective and disjoint -- and
\emph{acyclicity} -- we can automatically disprove equalities of the
form \(x = t\) where \(x\) appears constructor-guarded in \(t\).

\Spacedcommand{\CaseDesc}{\Function{case}}

However, since this paper works on the syntactic form of datatype
definitions, it is rife with ``\ldots'' definitions. For instance, the
authors reduce case analysis to induction with no less than ten
ellipsis in the construction. In our system, we generically derive
case analysis by a mere definition \emph{within the type theory}:
\[
  \Let{\CaseDesc 
        \PiTel{\Var{D}}{\Desc} 
        \PiTel{\Var{P}}{\Mu{\Var{D}} \To \Set} 
        \PiTel{\Var{cases}}{(\PiTo{\Var{d}}{\InterpretDesc{\Var{D}}[(\Mu{\Var{D}})]}
                               \Var{P} (\In{\Var{d}}))} 
        \PiTel{\Var{x}}{\Mu{\Var{D}}}}
      {\Var{P} \Var{x}}{
\CaseDesc[\Var{D}\:
          \Var{P}\:
          \Var{cases}\:
          \Var{x}] \DoReturn
       {\induction[\Var{D}\: \Var{P}\: (\Lam{\Var{d}}\Lam{\_}\Var{cases}\: \Var{d})\: \Var{x}]}
}
\]

\Spacedcommand{\NoConfusion}{\Function{NoConfusion}}
\Spacedcommand{\noConfusion}{\Function{noConfusion}}
\Spacedcommand{\DecideEqEnum}{\Function{decideEq-EnumT}}
\Spacedcommand{\DecideEqEqual}{\Constructor{equal}}
\Spacedcommand{\DecideEqNotEqual}{\Constructor{not-equal}}
\Spacedcommand{\DescEq}{\Function{DescEq}}

Similarly, the authors specify and prove the no confusion lemma over
the skeleton of an inductive definition. In our system, this result is
internalized through two definitions. In the following, we will assume
that \(D\) is a tagged description, \ie \(D = \Dsigma[\Meta{E}\:
  \Meta{T}]\) where \(\Meta{E}\) is the finite sets of constructor
labels. An inhabitant of \(\Mu[D]\) is therefore a pair
\(\In[\Pair{c}{a}]\) with \(c\) representing the constructor name and
\(a\) the tuple of arguments. We state the \(\NoConfusion\) lemma over
such code:
\[
\Let{\NoConfusion 
       & \PiTel{\Var{x}}{\Mu{\Var{D}}}
       & \PiTel{\Var{y}}{\Mu{\Var{D}}}}
    {\Set[1]}{
\With{\NoConfusion 
         & (\In[\Pair{\Var{x}}{\Var{a_x}}])
         & (\In[\Pair{\Var{y}}{\Var{a_y}}])}
     {\DecideEqEnum[\Var{x}\: \Var{y}]}{
\Return{\NoConfusion 
  & (\In[\Pair{\Var{x}}{\Var{a_x}}]) 
  & (\In[\Pair{\Var{x}}{\Var{a_y}}]) 
  & \WithArg{\DecideEqEqual[\Refl{}{}]}}
       {\PiTo{\Var{P}}{\Set} (\Var{a_x} \PropEqual \Var{a_y} \To \Var{P}) \To \Var{P}}
\Return{\NoConfusion
  & (\In[\Pair{\Var{x}}{\Var{a_x}}]) 
  & (\In[\Pair{\Var{y}}{\Var{a_y}}]) 
  & \WithArg{\DecideEqNotEqual[\Var{q}]}}
       {\PiTo{\Var{P}}{\Set} \Var{P}}

}}
\]
and then prove it by deciding the equality of enumeration index to
discriminate constructor names:
\[
\Let{\noConfusion 
      & \PiTel{\Var{x}}{\Mu{\Var{D}}}
      & \PiTel{\Var{y}}{\Mu{\Var{D}}}
      & \PiTel{\Var{q}}{\Var{x} \PropEqual \Var{y}}}
    {\NoConfusion[\Var{x}\: \Var{y}]}{
\With{\noConfusion 
         & (\In[\Pair{\Var{x}}{\Var{a_x}}])
         & (\In[\Pair{\Var{y}}{\Var{a_y}}])
         & \Var{q}}
     {\DecideEqEnum[\Var{x}\: \Var{y}]}{
\Return{\noConfusion 
  & (\In[\Pair{\Var{x}}{\Var{a_x}}]) 
  & (\In[\Pair{\Var{x}}{\Var{a_y}}]) 
  & \Var{q}
  & \WithArg{\DecideEqEqual[\Refl{}{}]}}
       {\Lam{\Var{P}}{\Lam{\Var{rec}}{\Var{rec}\: \Var{q}}}}
\Return{\noConfusion
  & (\In[\Pair{\Var{x}}{\Var{a_x}}]) 
  & (\In[\Pair{\Var{y}}{\Var{a_y}}]) 
  & \Var{q}
  & \WithArg{\DecideEqNotEqual[\Var{neq}]}}
       {\Lam{\Var{P}}{\EmptyElim[(\Var{neq}\: \Var{q})]}}

}}
\]
At this stage, we have proved this lemma generically, for all tagged
descriptions. Hence, after having defined a new datatype, a user can
directly use this lemma on her definition. For convenience, a
subsequent elaboration phase could also specialize such lemma to the
particular definition.

%% ** <- Generic deriving mechanism
%% *** <- Semidecidable procedure to check derivability
%% *** <- Generic operation parameterized by semidecision
%% *** -> Reflection of inductive
%% *** -> Example
%% **** <- Decidable equality

\subsection{Deriving operations on datatypes}

Another possible extension of our system is a generic
\texttt{deriving} mechanism. In the \textsc{Haskell} language, we can
write a definition such as
\[\Code{
\NatDef \\
\qquad\Kwd{deriving}\: \mathsf{Eq}
}\]
that automatically generates an equality test for the given
datatype. Again, since datatypes are a meta-theoretical entity, this
deriving mechanism has to be provided by the implementer and, template
programming aside, they cannot be implemented by the programmers
themselves. 

\Spacedcommand{\Derivable}{\Canonical{Derivable}}
\Spacedcommand{\Decidable}{\Canonical{Decidable}}
\Spacedcommand{\SemiDecidable}{\Canonical{SemiDecidable}}
\newcommand{\subDesc}{\CN{subDesc}}
\newcommand{\decideIn}{\CN{membership}}
\newcommand{\derive}{\CN{derive}}

In our framework, we could extend the elaborator for datatypes with a
\texttt{deriving} mechanism. However, for such a mechanism to work, we
must restrict ourselves to decidable properties: for example, if the
user asks to derive equality on a datatype that do not have a
decidable equality (\eg, Brouwer ordinals), the system should fail
immediately. To solve this issue, we add one level of
indirection\if 0 \footnote{Quoting David Wheeler, "all problems in computer
  science can be solved by another level of indirection"}\fi: while we
cannot decide equality for \emph{any} datatype, we can decide whether
the datatype belongs to a sub-universe \emph{for which} equality is
decidable. Hence, to introduce a derivable property \(P\) in the type
theory, the programmer would populate the following record structure:
\[
\Let{\Derivable & \PiTel{\Var{P}}{\Desc \To \Set}}
    {\Set[1]}{
\multicolumn{5}{@{}l}{
\Derivable[\Var{P}] \DoReturn
\left\{
\begin{array}{l@{\:}l}
  \TypeAnn{\Var{\subDesc} &}{\Desc \To \Set[1]} \\
  \TypeAnn{\Var{\decideIn} &}{\PiTo{\Var{D}}{\Desc} \Decidable[(\Var{\subDesc}\: \Var{D})]} \\
  \TypeAnn{\Var{\derive} &}{\Var{\subDesc}\: \Var{D} \To \Var{P}\: \Var{D}}
\end{array}
\right.
}}
\]

\newcommand{\eqDesc}{\Function{eqDesc}}
\newcommand{\decideInEq}{\Function{membershipEq}}
\newcommand{\deriveEq}{\Function{deriveEq}}

For example, in the case of equality, the programmer has first to
provide a function \(\TypeAnn{\eqDesc}{\Desc \To \Set[1]}\). One
possible (perhaps simplistic, but valid) sub-universe consists only of
products, finite sums, recursive call, and unit: it is enough to describe
natural numbers and variants thereof. She then implements a procedure
\(\TypeAnn{\decideInEq}{\PiTo{\Var{D}}{\Desc}
  \Decidable[(\Var{\eqDesc}\: \Var{D})]}\) deciding whether a given
\(\Desc\) code fits into this sub-universe or not. Recall that
\(\Decidable[\Meta{A}]\) corresponds to \(\Meta{A} \Sum
\Neg[\Meta{A}]\). It should be clear that the membership of a
\(\Desc\) code to our sub-universe of finite products and sums is
decidable. Finally, she implements the key operation
\(\TypeAnn{\deriveEq}{\eqDesc\: \Var{D} \To \PiTo{\Var{x}\:
    \Var{y}}{\Mu[\Var{D}]}{\Decidable[(\Var{x} \PropEqual \Var{y})]}}\) that
decides equality of two objects, assuming that they belong to the
sub-universe. This implements the structure
\(\TypeAnn{\Canonical{Eq}}
          {\Derivable[(\Lam{\Var{D}}
                          {\PiTo{\Var{x}\:\Var{y}}{\Mu[\Var{D}]}
                             \Decidable[\Var{x} \PropEqual \Var{y}]})]}\).

\Spacedcommand{\NatEq}{\Function{Nat-eq}}
\Spacedcommand{\Witness}{\Function{witness}}

While elaborating a datatype, it is then straightforward -- and
automatic -- for us to generate its derivable property, or reject it
immediately: we simply compute \(\decideIn\) on the specific code. If
we obtain a negative response, we report an error. If we obtain a
positive witness, we pass that witness to \(\derive\) and instantiate
the derived property. 

%\clearpage
For example, since natural numbers fit into the \(\eqDesc\) universe,
the elaboration machinery would automatically generate the following
decision procedure
\[
\Let{\NatEq & \PiTel{\Var{x}\: \Var{y}}{\Nat}}
    {\Decidable[\Var{x} \PropEqual \Var{y}]}{
\Return{\NatEq & \Var{x}\: \Var{y}}
       {\deriveEq\: (\Witness[(\decideInEq\: \NatD)\: \Void])}
}
\]
without any input from the user but the \texttt{deriving}
\(\Canonical{Eq}\) clause. Here, \(\Witness\) is a library function
that extracts the witness from a true decidable property: applied to
\(\NatD\), the function \(\decideInEq\) computes a positive witness
that we can simply extract.

%% * Conclusion

\newpage
\section{Conclusion}

%% ** <- Machinery for type-directed elaboration
%% *** <- Using type as representations
%% *** -> Guiding elaboration task
%% *** -> Conceptually simple
%% *** -> Stronger soundness properties
%% ** <- Reduction of data-syntax to code
%% *** <- Semantics-based approach to inductives
%% *** -> No more need of dotdotdot proofs
%% *** -> Reflection of constructions on datatypes
%% **** <- Construction on constructors
%% **** <- Deriving mechanism

In this paper, we have striven to give a coherent framework for
elaboration in type theory. We have organized our system around two
structuring ideas. First, by using the flow of typing information, we
obtain a richer and more flexible term language. Second, by using
types as presentation of high-level concepts, such as inductive
definition, we can effectively guide the elaboration process. This
technique is conceptually simple and therefore amenable to formal
reasoning. This simplicity together with the soundness proofs should
convince the reader of its validity.

%% ** <- Making "treatment of datatypes technically straightforward"
%% *** <- Often "conceptually straightforward"
%% **** <- Paraphrasing Stone-Harper
%% *** <- Elaboration formalized once and for all
%% *** <- Results reusable across calculis
%% *** -> Construction can be described on codes
%% **** <- Internalized syntax of datatypes
%% **** -> Subject to formal verification, too

From there, we believe that reasoning on inductive definitions can be
liberated from the elusive ellipsis: proofs and constructions on
inductives ought to happen within the type theory itself. After
\citet{harper:elaboration}, we claim that if the treatment of
datatypes is conceptually straightforward then it ought to be
technically straightforward and implemented as a generic program in
the type theory. For non-straightforward properties, our results
should be reusable across calculi -- such as the Calculus of Inductive
Constructions -- and not too rigidly tied to our universe of
datatypes. Besides, we were careful to present elaboration as a
relation rather than a mere program, making it more amenable to
abstract reasoning.

%% ** <- Illustrated possibilities of such system
%% *** -> Bootstrapping inductive fragment
%% **** -> Generic theorems for inductives
%% **** -> Safety of correctness
%% ***** <- Internalized, hence type-checked
%% *** -> Native support for generic programming
%% **** -> Deriving mechanism
%% **** -> No extension to type theory
%% ***** <- Merely an elaboration task

To support our claim that inductives should be bootstrapped, we have
presented two possible extensions to the elaboration process. While we
did not formalize these examples, our expectations seem reasonable and
our experience modeling them in Agda further supports this
impression. We have seen how generic theorems on inductive types can
be internalized as generic programs: besides the benefit of reducing
the trusted computing base, their validity is guaranteed by type
checking. Also, we have glimpsed at a generic \texttt{deriving}
mechanism: with no extension to the type theory, we are able to
let the user define sub-universes that support certain
operations. These operations could then be automatically specialized
to the datatypes that support them, all that without any user
intervention.

%% ** <- Future work
%% *** <- Adapt to coinductive recursion
%% *** <- Capture Induction-Recursion
%% *** <- Formalize deriving
%% *** <- Implementation
%% **** <- Generate specialized induction principles, lemmas
%% **** <- Internalization in Type Theory

\paragraph{Future work:} A next step would be to adapt our syntax to coinductive
types. Similarly, it would be interesting to see if it scales to
inductive-recursive definitions. On the implementation side, we are
currently implementing these various systems in a toy type
theory. This step will require a formalization of our
\texttt{deriving} mechanism. Also, we will have to generate the
specialized induction principles (such as case analysis and
course-of-value recursion) as well as the constructions on
constructors. Finally, an interesting challenge would be to
internalize the elaboration process \emph{itself} in type theory,
hence obtaining a correct-by-construction translation.

\paragraph{Acknowledgements} 
We would like to thank Pierre Boutillier for pointing us to the
relevant literature on Coq's treatment of inductives. We also thank
our colleagues Guillaume Allais and Stevan Andjelkovic for many
stimulating discussions and for their input on this paper. The authors
are supported by the Engineering and Physical Sciences Research
Council, Grant EP/G034699/1.

%% * Biblio

\newpage
\bibliographystyle{abbrvnat}
\bibliography{paper,../thesis-2011-phd/levitation,../thesis-2011-phd/funorn}

%% Local Variables:
%% mode: outline-minor
%% outline-regexp: "%% [*\f]+"
%% outline-level: outline-level
%% End:       
               
\end{document}